\documentclass[11pt]{article}%
\usepackage{fullpage}
\usepackage{amssymb,amsmath}
\usepackage{graphicx, epsfig}

\def\id#1{\ensuremath{\mathit{#1}}}
\let\idit=\id
\def\idbf#1{\ensuremath{\mathbf{#1}}}
\def\idrm#1{\ensuremath{\mathrm{#1}}}
\def\idtt#1{\ensuremath{\mathtt{#1}}}
\def\idsf#1{\ensuremath{\mathsf{#1}}}
\def\idcal#1{\ensuremath{\mathcal{#1}}}  

\def\floor#1{\lfloor #1 \rfloor}
\def\ceil#1{\lceil #1 \rceil}
\def\etal{\emph{et~al.}}

\newcommand{\no}[1]{}
\newcommand{\myparagraph}[1]{{\bf #1}}
\newcommand{\todo}[1]{} 
\newcommand{\settwo}[2]{\ensuremath{\{\,#1\,|\,#2\,\} }}

\newtheorem{theorem}{Theorem}
\newtheorem{lemma}{Lemma}
\newenvironment{proof}{\trivlist\item[]\emph{Proof}:}%
{\unskip\nobreak\hskip 1em plus 1fil\nobreak$\Box$
\parfillskip=0pt%
\endtrivlist}

\newenvironment{itemize*}%
  {\begin{itemize}%
    \setlength{\itemsep}{0pt}%
    \setlength{\parskip}{0pt}%
    \setlength{\parsep}{0pt}%
    \setlength{\topsep}{0pt}%
    \setlength{\partopsep}{0pt}%
  }%
  {\end{itemize}}%

\newcommand{\cT}{{\cal T}}
\newcommand{\tL}{{\tilde L}}
\newcommand{\tS}{{\tilde S}}
\newcommand{\tB}{{\tilde B}}
\newcommand{\tE}{{\tilde E}}
\newcommand{\tW}{{\tilde W}}
\newcommand{\cTb}{{T}^{bwt}}
\newcommand{\cD}{{\cal D}}
\newcommand{\cB}{{\cal B}}
\newcommand{\cK}{{\cal K}}
\newcommand{\cC}{{\cal C}}
\newcommand{\cI}{{\cal I}}
\newcommand{\cL}{{\cal L}}
\newcommand{\cN}{{\cal N}}
\newcommand{\cM}{{\cal M}}
\newcommand{\Temp}{{Temp}}
\newcommand{\bT}{\mathbb{T}}
\newcommand{\tR}{{\widetilde R}}
\newcommand{\oX}{{\overline X}}
\newcommand{\oY}{{\overline Y}}
\newcommand{\oZ}{{\overline Z}}
\newcommand{\oS}{{\overline S}}
\newcommand{\on}{{\overline n}}
\newcommand{\tpi}{{\widetilde \pi}}
\newcommand{\eps}{\varepsilon}
\newcommand{\tpred}{\mathrm{tpred}}
\newcommand{\Col}{\mathrm{Col}}
\newcommand{\occ}{\mathrm{occ}}

\newcommand{\tcount}{t_{\mathrm{count}}}
\newcommand{\trange}{t_{\mathrm{range}}}
\newcommand{\tlocate}{t_{\mathrm{locate}}}
\newcommand{\textract}{t_{\mathrm{extract}}}
\newcommand{\tsa}{t_{\mathrm{SA}}}

\newcommand{\ra}{\idrm{rank}}
\newcommand{\sel}{\idrm{select}}
\newcommand{\acc}{\idrm{access}}
\newcommand{\init}{\idrm{init}}
\newcommand{\Pref}{\mathit{Pref}}
\newcommand{\shortver}[1]{}
\newcommand{\longver}[1]{#1}
\newcommand{\shlongver}[2]{#2}

\begin{document}
\title{Compressed Data Structures  for  Dynamic Sequences}

 \author{
 J. Ian Munro\thanks{Cheriton School of Computer Science, University of Waterloo. Email {\tt imunro@uwaterloo.ca}.}
 \and
 Yakov Nekrich\thanks{Cheriton School of Computer Science, University of Waterloo.
 Email: {\tt yakov.nekrich@googlemail.com}.}
}
\date{}
\maketitle

\begin{abstract}
We consider the problem of storing a dynamic string $S$ over an alphabet $\Sigma=\{\,1,\ldots,\sigma\,\}$ in compressed form. Our representation supports insertions and deletions of symbols and answers three 
fundamental queries: $\acc(i,S)$ returns the $i$-th symbol in $S$, $\ra_a(i,S)$
counts how many times a symbol $a$ occurs among the first $i$ positions in $S$, 
and $\sel_a(i,S)$ finds the position where a symbol $a$ occurs for the $i$-th time. We present the first fully-dynamic data structure for arbitrarily large alphabets that achieves optimal query times for all three operations and supports updates with worst-case time guarantees. Ours is also the first fully-dynamic data structure that needs only $nH_k+o(n\log\sigma)$ bits, where $H_k$ is the $k$-th order entropy and $n$ is the string length. 
Moreover our representation supports extraction of a substring $S[i..i+\ell]$ in optimal $O(\log n/\log\log n + \ell/\log_{\sigma}n)$ time.
\end{abstract}
 

\section{Introduction}
In this paper we consider the problem of storing a sequence $S$ of length $n$ over an alphabet $\Sigma=\{\,1,\ldots,\sigma\,\}$ so that the following operations are supported:\\
- $\acc(i,S)$ returns the $i$-th symbol, $S[i]$, in $S$\\
- $\ra_a(i,S)$ counts how many times $a$ occurs among the first $i$ symbols 
in $S$, $\ra_a(i,S)=|\{\,j\,|\, S[j]=a \text{ and } 1\le j\le i\,\}|$\\
-$\sel_a(i,S)$ finds the position in $S$ where $a$ occurs for the $i$-th time,
$\sel_a(i,S)=j$ where $j$ is such that  $S[j]=a$ and $\ra_a(j,S)=i$. \\
This problem, also known as the rank-select problem, is one of the most fundamental problems in compressed data structures. 
There are many  data structures that store a string in compressed form and support three above defined operations efficiently. There are static data structures that use $nH_0+o(n\log\sigma)$ bits 
or even $nH_k+o(n\log \sigma)$ bits for any $k\le  \alpha\log_{\sigma}n-1$  and a positive constant $\alpha<1$\footnote{Henceforth $H_0(S)=\sum_{a\in \Sigma}\frac{n_a}{n}\log\frac{n}{n_a}$, where $n_a$ is the number of times $a$ occurs in $S$, 
is the $0$-th order entropy and  $H_k(S)$ for $k\ge 0$ is the $k$-th order empirical entropy. $H_k(S)$ can be defined as 
$H_k(S)=\sum_{A\in \Sigma^k}|S_A|H_0(S_A)$, where $S_A$ is the subsequence of $S$ generated by symbols that follow the $k$-tuple $A$; $H_k(S)$ is the lower bound on the average space
usage of any statistical compression method that encodes each symbol using the context of $k$ previous symbols~\cite{Manzini01}.}. Efficient static rank-select data structures are described in \cite{GGV03,GMR06,FMMN07,LP07,LP09,BGNN10,HM10,NS14,BN12}. We refer to~\cite{BN12} for most recent results and a discussion of previous static solutions.

In many situations we must work with dynamic sequences. We must be able to insert a new symbol at an arbitrary 
position $i$ in the sequence or delete an arbitrary symbol $S[i]$.  The design of dynamic solutions, that support insertions and deletions of symbols, is an important problem.  Fully-dynamic data structures for rank-select problem were considered in~\cite{HSS03,CHL04,BB04,MN06,CHLS07,GHSV07,MN08,HSS11}. 
Recently Navarro and Nekrich~\cite{NavarroN13,NavarroN13a} obtained a fully-dynamic solution with $O(\log n/\log \log n)$ times for $\ra$, $\acc$, and $\sel$ operations. By the lower bound of Fredman and Saks~\cite{FS89}, these query times are optimal.  The data structure described in~\cite{NavarroN13} uses $nH_0(S)+o(n\log\sigma)$ bits and supports updates in $O(\log n/\log \log n)$ amortized time. It is also possible to support updates in $O(\log n)$ worst-case time, but then the time for answering a $\ra$ query grows to $O(\log n)$~\cite{NavarroN13a}. All previously known fully-dynamic data structures need at least $nH_0(S)+o(n\log\sigma)$ bits. Two only exceptions are data structures of Jansson et al.~\cite{JanssonSS12} and Grossi et al.~\cite{GrossiRRV13} that keep $S$ in $nH_k(S)+o(n\log\sigma)$ bits, but do not support $\ra$ and $\sel$ queries. 
A more restrictive dynamic scenario was considered by Grossi et al.~\cite{GrossiRRV13} and Jansson et al.~\cite{JanssonSS12}: an update operation \emph{replaces} a symbol $S[i]$ with another symbol so that the total length of $S$ does not change, but insertions of new symbols or deletions of symbols of $S$ are not supported.  Their data structures need $nH_k(S)+ o(n \log\sigma)$ bits and answer access queries in $O(1)$ time; the data structure of Grossi et al.~\cite{GrossiRRV13} also supports rank and select queries in $O(\log n/\log\log n)$ time.

In this paper we describe the first fully-dynamic data structure 
that keeps the input sequence in $nH_k(S)+o(n\log\sigma)$ bits; our representation supports $\ra$, $\sel$, and $\acc$ queries 
in optimal $O(\log n/\log \log n)$ time. 
Symbol insertions and deletions at any position in $S$ are supported in $O(\log n/\log \log n)$ worst-case time.
We list our and previous results for fully-dynamic sequences in Table~\ref{tab:ranksel}. 
Our representation of dynamic sequences also supports the operation of extracting a substring. Previous dynamic data structures 
require $O(\ell)$ calls of $\acc$ operation in order to extract the substring of length $\ell$. Thus the previous best fully-dynamic 
representation, described in~\cite{NavarroN13} needs $O(\ell(\log n/\log\log n))$ time to extract a substring $S[i..i+\ell-1]$ of $S$. Data structures described in~\cite{GrossiRRV13} and~\cite{JanssonSS12} support substring extraction in $O(\log n/\log\log n+\ell/\log_{\sigma}n)$ time but they either do not support $\ra$ and $\sel$ queries or they support only updates that replace a symbol with another symbol. Our dynamic data structure can extract a substring in optimal $O(\log n/\log\log n+\ell/\log_{\sigma}n)$ time without any restrictions on updates or queries. 


\begin{table}[tb]
  \centering
\resizebox{\textwidth}{!}{
  \begin{tabular}{|l|c|c|c|c|c|c|} \hline
    Ref. & Space & Rank & Select & Access & Insert/  &  \\
         &       &      &        &        & Delete   &  \\ \hline \hline
\cite{HM10} & $nH_0(S)+o(n\log\sigma)$ & \multicolumn{3}{|c| }{$O((1+\log\sigma/\log\log n)\lambda)$} & $O((1+\log\sigma/\log\log n)\lambda)$ & W \\ \hline
\cite{NS14} & $nH_0(S)+o(n\log\sigma)$ & \multicolumn{3}{|c|}{$O((\log\sigma/\log\log n)\lambda)$} & $O((\log\sigma/\log\log n)\lambda)$ & W \\
\hline
\cite{NavarroN13}   & $nH_0(S)+o(n\log\sigma)$ & $O(\lambda)$ & $O(\lambda)$ & $O(\lambda)$ & $O(\lambda)$ & A \\ \hline
\cite{NavarroN13}   & $nH_0(S)+o(n\log\sigma)$ & $O(\log n)$ & $O(\lambda)$ & $O(\lambda)$ & $O(\log n)$ & W \\   
\hline
\cite{JanssonSS12} & $nH_k+ o(n\log\sigma) $& - & - & $O(\lambda)$ & $O(\lambda)$ & W \\ \hline
\cite{GrossiRRV13} & $nH_k+ o(n\log\sigma)$ & - & -  &  $O(\lambda)$ &  $O(\lambda)$ & W \\ \hline\hline
New                & $nH_k+ o(n\log\sigma)$ & $O(\lambda)$ & $O(\lambda)$ & $O(\lambda)$ & $O(\lambda)$  & W \\ \hline \hline
  \end{tabular}
}
  \caption{Previous and New Results for Fully-Dynamic Sequences. The rightmost column indicates whether updates are amortized (A) or worst-case (W). We use notation $\lambda=\log n/\log\log n$ in this table.}
  \label{tab:ranksel}
\end{table}


In Section~\ref{sec:ranklog} we describe a data structure that uses $O(\log n)$ bits per symbol and supports $\ra$, $\sel$, and $\acc$ in optimal $O(\log n/\log\log n)$ time. This data structure essentially maintains a linked list $L$ containing 
all symbols of $S$; using some auxiliary data structures on $L$, we can answer  $\ra$, $\sel$, and $\acc$ queries on $S$. 
In Section~\ref{sec:ranklogsigma} we show how the space usage can be reduced to $O(\log \sigma)$ bits per symbol. A compressed data structure that needs $H_0(S)$ bits per symbol is presented in Section~\ref{sec:compr}.  The approach of Section~\ref{sec:compr} is based on dividing $S$ into a number of subsequences. We store a fully-dynamic data structure for only one such subsequence of appropriately small size.  Updates on other subsequences are supported by periodic re-building.
 In Section~\ref{sec:compr2} we show that the space usage can be reduced to $nH_k(S)+o(n\log\sigma)$. 

\section{$O(n\log n)$-Bit Data Structure}
\label{sec:ranklog}
We start by describing a data structure that uses $O(\log n)$ bits per symbol.
\begin{lemma}
  \label{lemma:logsigma0}
A dynamic string $S[1,m]$ for $m\le n$ over alphabet $\Sigma=\{\,1,\ldots,\sigma\,\}$ 
can be stored in a data structure that needs
$O(m\log m)$ bits, and answers queries
$\acc$, $\ra$ and $\sel$ in time $O(\log m/\log \log n)$.
Insertions and deletions of symbols are supported in $O(\log m/\log \log n)$ time. The data structure uses a universal look-up table of size $o(n^{\eps})$ for an arbitrarily small $\eps>0$.
\end{lemma}
\begin{proof}
  We keep elements of $S$ in a list $L$. Each  entry of $L$ contains a symbol $a\in \Sigma$. 
For every $a\in \Sigma$, we also maintain the list $L_a$. Entries of $L_a$ correspond to those entries of $L$ that contain the symbol $a$.
We maintain data structures $D(L)$ and $D(L_a)$ that enable us to find the number 
of entries in $L$ (or in some list $L_a$) that precede 
an entry $e\in L$ (resp.\ $e\in L_a$); we can also find the $i$-th entry $e$ in $L_a$ or $L$ using $D(L_{\cdot})$.
We will prove in Lemma~\ref{lemma:dl} that $D(L)$ needs $O(m\log m)$ bits and supports queries and updates on $L$ in $O(\log m/\log \log n)$ time.

We can answer a query $\sel_a(i,S)$ by finding the $i$-th entry $e_i$ in $L_a$, following the pointer from $e_i$ to the corresponding entry $e'\in L$, and counting the number $v$  of entries preceding $e'$ in $L$. Clearly\footnote{To simplify the description, we assume that a list entry precedes itself.}, $\sel_a(i,S)=v$. To answer a query $\ra_a(i,S)$, we first find the $i$-th entry $e$ in $L$. Then we find  the last entry $e_a$  that precedes $e$ and contains $a$. Such queries can be answered in $O((\log \log \sigma)^2\log \log m)$ time as will be shown in \shlongver{the full version of this paper, attached at the end of this submission}{Lemma~\ref{lemma:colpred} in Section~\ref{sec:markpred}}. If $e'_a$ is the entry that corresponds to $e_a$ in $L_a$, then $\ra_a(i,S)=v$, where $v$ is the number of entries that precede $e_a'$ in $L_a$. 
\end{proof}

\section{$O(n\log \sigma)$-Bit Data Structure}
\label{sec:ranklogsigma}
\begin{lemma}
  \label{lemma:logsigma}
A dynamic string $S[1,n]$ over alphabet $\Sigma=\{\,1,\ldots,\sigma\,\}$ 
can be stored in a data structure using 
$O(n\log\sigma)$ bits, and supporting queries
$\acc$, $\ra$ and $\sel$ in time $O(\log n/\log \log n)$.
Insertions and deletions of symbols 
are supported in $O(\log n/\log \log n)$ time. 
\end{lemma}
\begin{proof}
If $\sigma=\log^{O(1)}n$, then the data structures described in~\cite{NS14} and~\cite{HM10} provide desired query and update times. The case $\sigma=\log^{\Omega(1)}n$ is considered below. 

\tolerance=1000
We show how the problem on a sequence of size $n$ can be reduced to the same problem on a sequence of size $O(\sigma\log n)$. 
The sequence $S$ is divided into chunks.
We can maintain the size $n_i$ of each chunk $C_i$, so that $n_i=O(\sigma\log n)$ and the total number of chunks is bounded by $O(n/\sigma)$. We will show how to maintain chunks in \shlongver{the full version of this paper}{Section~\ref{sec:logsigmaupd}}. 
For each
$a\in \Sigma$, we keep a global bit sequence $B_a$. 
$B_a=1^{d_1}01^{d_2}0\ldots1^{d_i}0\ldots$ where $d_i$ is the number 
of times $a$ occurs in the chunk $C_i$. 
We also keep a bit sequence $B_t=1^{n_1}01^{n_2}0\ldots 1^{n_i}0\ldots$.
We can compute 
$\ra_a(i,S)=v_1+v_2$ where $v_1=\ra_1(\sel_0(j_1,B_a),B_a)$, $j_1= \ra_0(\sel_1(i,B_t),B_t)$, $v_2=\ra_a(i_1,C_{i_2})$,
$i_2=j_1+1$ and $i_1=i-\ra_1(\sel_0(j_1,B_t),B_t)$.     
To answer a query $\sel_a(i,S)$, we first find the index $i_2$
of the chunk $C_{i_2}$ that contains the $i$-th occurrence of $i$,
$i_2=\ra_0(\sel_1(i,B_a),B_a)+1$. Then we find $v_a=\sel_a(C_{i_2},i-i_1)$ for $i_1=\ra_1(\sel_0(i_2-1,B_a),B_a)$;
$v_a$ identifies the position of the $(i-i_1)$-th occurrence of $a$ in the chunk $C_{i_2}$, where $i_1$ denotes the number of $a$'s in the first $i_2-1$ chunks. Finally we compute $\sel_a(i,S)=v_a+s_p$ where $s_p=\ra_1(\sel_0(i_2-1,B_t),B_t)$ 
is the total number of symbols in the first $i_2-1$ chunks. 
We can support queries and updates on $B_t$ and on each $B_a$ 
in $O(\log n/\log \log n)$ time~\cite{NS14}. 
By Lemma~\ref{lemma:logsigma0}, queries and updates on $C_i$ 
are supported in $O(\log \sigma/\log \log n)$ time. 
Hence, the query and update times of our data structure are 
$O(\log n/\log \log n)$.

 $B_t$ can be kept in $O((n/\sigma)\log\sigma)$ bits~\cite{NS14}. The array 
$B_a$ uses $O(n_a\log\frac{n}{n_a})$ bits, where $n_a$ is the number of times $a$ occurs in $S$.  \no{$O((n/\sigma)\log\frac{n\sigma}{n_a})=O((n/\sigma)(\log\frac{n}{n_a}+\log\sigma))$ bits. Hence all $B_a$ need $O(n\log \sigma)
+O(\sum_a(n/\sigma)\log n/n_a)$. By Jensen's inequality, the second term of this expression can be bounded by $O(n\log \sigma)$.} Hence,  all $B_a$ and $B_t$ use $O((n/\sigma)\log\sigma+\sum_an_a\log\frac{n}{n_a})=O(n\log\sigma)$ bits.
By Lemma~\ref{lemma:logsigma0}, we can also keep the data structure for each chunk in $O(\log \sigma+\log\log n)=O(\log\sigma)$ bits per symbol.
\end{proof}

\section{Compressed Data Structure}
\label{sec:compr}
In this Section we describe a data structure that uses $H_0(S)$ bits per symbol. We start by considering the case when the alphabet size is not too large, $\sigma\le n/\log^3 n$.
The sequence $S$ is split into subsequences $S_0$, $S_1$, $\ldots$ $S_r$ for $r=O(\log n/(\log\log n))$. 
The subsequence $S_0$ is stored in $O(\log \sigma)$ bits per element as described in Lemma~\ref{lemma:logsigma}. 
Subsequences $S_1,\ldots S_r$ are substrings of $S\setminus S_0$. $S_1,\ldots S_r$ are stored in compressed static data structures. New elements are always inserted into the subsequence $S_0$.  Deletions from $S_i$, $i\ge 1$, are implemented as lazy deletions: an element in $S_i$ is marked as deleted. We guarantee that the number of elements that are marked as deleted  is bounded by 
$O(n/r)$. If a subsequence $S_i$ contains many elements marked as deleted, it is re-built: we create a new instance of $S_i$ 
that does not contain deleted symbols. If a symbol sequence $S_0$ contains too many elements, we insert the elements of $S_0$ into 
$S_i$ and re-build $S_i$ for $i\ge 1$.
 Processes of constructing a new subsequence and re-building a subsequence with too many obsolete elements are run 
in the background.

\no{As follows from the above high-level description, $S_i$ are subsequences (but not substrings!) of $S$: if an element $e_1\in S_1$ precedes $e_2\in S_2$, then $e_1$ precedes $e_2$ in $S$. But the index of the subsequence $S_j$ that contains an element $e$ 
depends only on the time when it was inserted and does not depend on its position in $S$. Thus  subsequences can be interspersed in an arbitrary way. }

The bit sequence $M$ identifies elements in $S$ that are marked as deleted: $M[j]=0$ if and only if  $S[j]$ is marked as deleted. 
The bit sequence $R$ distinguishes between the elements of $S_0$ and elements of $S_i$, $i\ge 1$: $R[j]=0$ if the $j$-th element of $S$ is kept in $S_0$ and $R[j]=1$ otherwise. 

We further need  auxiliary data structures for answering 
select queries. We start by defining an auxiliary subsequence $\tS$ that contains copies of elements already stored in other subsequences.  
Consider a subsequence $\oS$ obtained by merging subsequences $S_1$, $\ldots$, $S_r$ (in other words, $\oS$ is obtained from $S$ by removing elements of $S_0$). 
Let $S'_a$ be the subsequence obtained by selecting (roughly) every $r$-th occurrence of a symbol $a$ in $\oS$. 
The subsequence $S'$ is obtained by merging  subsequences $S'_a$ for all $a\in \Sigma$.  
Finally $\tS$ is obtained  by merging $S'$ and $S_0$.  
 We support queries $\sel'_a(i,\tS)$ on $\tS$, defined as follows: $\sel'_a(i,\tS)=j$ such that (i) a copy of $S[j]$ is stored in $\tS$ and (ii) if $\sel_a(i,S)=j_1$, then $j\le j_1$ and copies of elements $S[j+1]$, $S[j+2]$, $\ldots$, $S[j_1]$ are not stored in $\tS$. That is, $\sel'_a(i,\tS)$ returns the largest index  $j$, such that  $S[j]$ precedes $S[\sel_a(i,S)]$ and $S[j]$ is also stored in $\tS$. The data structure for $\tS$ delivers approximate 
answers for $\sel$ queries; we will show later how the answer to a query $\sel_a(i,S)$ can be found quickly if the answer to $\sel'_a(i,\tS)$ is known. Queries $\sel'(i,\tS)$ can be implemented using standard operations on a bit sequence of size $O((n/r)\log\log n)$  bits; for completeness, we provide a description in Section~\ref{sec:appselprime}.  
We remark that $\oS$ and $S'$ are introduced to define $\tS$; these two subsequences are not stored in our data structure. 
The bit sequence $\tE$ indicates what symbols of $S$ are also stored in $\tS$:
$\tE[i]=1$ if a copy of $S[i]$ is stored in $\tS$ and $\tE[i]=0$ otherwise. 
The bit sequence $\tB$ indicates what symbols in $\tS$ are actually from $S_0$: $\tB[i]=0$ iff $\tS[i]$ is stored in the subsequence $S_0$. 
Besides, we keep bit sequences $D_a$ for each $a\in \Sigma$. Bits of $D_a$ correspond to occurrences of $a$ in $S$. 
If the $l$-th occurrence of $a$ in $S$ is marked as deleted, 
then $D_a[l]=0$. All other bits in $D_a$ are set to $1$.

We provide the list of subsequences in Table~\ref{tab:auxdata}. Each subsequence is augmented with a data structure that supports rank and select queries.  
For simplicity we will not distinguish between a subsequence and a data structure on its elements. 
If a subsequence supports updates, then either (i) this is a subsequence over 
a small alphabet or (ii) this subsequence contains a small number of elements. 
In case (i), the subsequence is over an alphabet of constant size; by ~\cite{NS14,HM10} queries on such subsequences are answered in $O(\log n/\log \log n)$ time. In case (ii) the subsequence contains $O(n/r)$ elements; data structures on such subsequences are implemented as in Lemma~\ref{lemma:logsigma}. 
All auxiliary subsequences, except for $\tS$, are of type (i). Subsequence $S_0$ and an auxiliary subsequence $\tS$ are of type (ii). Subsequences $S_i$ for $i\ge 1$ are static, i.e. they are stored in data structures that do not support updates. We re-build these subsequences when they contain too many obsolete elements.  Thus dynamic subsequences support $\ra$, $\sel$, $\acc$, and updates in $O(\log n/\log\log n)$ time. It is known that we can implement all basic operations on a static sequence in $O(\log n/\log \log n)$ time\footnote{Static data structures also achieve significantly faster query times, but this is not necessary for our implementation.}. Our data structures on static subsequences are based on the approach of Barbay et al.~\cite{BHMR07}; however, our data structure can be constructed faster when the alphabet size is small and supports a substring extraction operation. A full description will be given in \shlongver{the full version of this paper}{Section~\ref{sec:construct}}. We will show below that queries on $S$ are answered by $O(1)$ queries on dynamic subsequences and $O(1)$ queries on static subsequences.

We also maintain arrays $Size[]$ and $Count_a[]$ for every $a\in \Sigma$. For any $1\le i \le r$, $Size[i]$ is the number of symbols in $S_i$ and $Count_a[i]$ specifies how many times $a$ occurs in $S_i$. We keep a data structure that computes  the sum of the first $i\le r$ entries in $Size[i]$ and find the largest $j$ such that $\sum_{t=1}^j Size[t]\le q$ for any integer $q$. The same kinds of queries are also supported on $Count_a[]$.  Arrays $Size[ ]$ and 
$Count_a[]$ use $O(\sigma\cdot r \cdot \log n)=O(n/\log n)$ bits.

\begin{table}[tb]
  \centering
  \begin{tabular}{|l|c|c|c|} \hline
    Name & Purpose & Alph. & Dynamic/ \\
         &          & Size  & Static\\ \hline
    $S_0$                & Subsequence of $S$ & - & Dynamic \\
    $S_i$, $1\le i\le r$ & Subsequence of $S$ & - & Static \\
    $M$       & Positions of symbols in $S_i$, $i\ge 1$, that are marked as deleted & const & Dynamic \\
    $R$                  & Positions of symbols from $S_0$ in $S$ \no{Relative order of symbols in $S_0$ and $S_i$ for $i\ge 1$} & const & Dynamic\\
    $\tS$                & Delivers an approximate answer to select queries          & -  & Dynamic \\
    $S'_a$, $a\in \Sigma$ & Auxiliary sequences for $\tS$                            & - & Dynamic \\
    $\tE$                & Positions of symbols from $\tS$ in $S$                    & const & Dynamic \\
    $\tB$                & Positions of symbols from $S_0$ in $\tS$  & const & Dynamic \\
    $D_a$                & Positions of symbols marked as deleted among all $a$'s    & const & Dynamic \\
\hline
  \end{tabular}
  \caption{Auxiliary subsequences for answering rank and select queries. A subsequence is dynamic if both insertions and deletions are supported. If a subsequence is static, then updates are not supported. Static subsequences are re-built when they contain too many obsolete elements.}
  \label{tab:auxdata}
\end{table}

\paragraph{Queries.}
To answer a query $\ra_a(i,S)$, we start by computing $i'=\sel_1(i,M)$; $i'$ is the position of the $i$-th element that is not marked as deleted. Then we find $i_0=\ra_0(i',R)$ and $i_1=\ra_1(i',R)$. By definition of $R$, $i_0$ is the number of elements of $S[1..i]$ that are stored in the subsequence $S_0$. 
The number of $a$'s in $S_0[1..i_0]$  is computed as $c_1=\ra_a(i_0,S_0)$.
The number of $a$'s in $S_1,\ldots, S_r$ before the position $i'$ is 
found as follows. We identify the index $t$, such that $\sum_{j=1}^t Size[j]<i_1\le \sum_{j=1}^{t+1}Size[j]$. 
Then we compute how many times $a$ occurred in $S_1,\ldots ,S_t$, $c_{2,1}=\sum_{j=1}^{t}Count_a[j]$, 
and in the relevant prefix of $S_{t+1}$, $c_{2,2}=\ra_a(i_1-\sum_{j=1}^{t}Size[j],S_{t+1})$. 
Let $c_2=\ra_1(c_{2,1}+c_{2,2},D_a)$. Thus $c_2$ is the number of symbols '$a$' that are not marked as deleted among the first $c_{2,1}+c_{2,2}$ occurrences of $a$ in $S\setminus S_0$.
Hence $\ra_a(i,S)=c_1+c_2$.

To answer a query $\sel_a(i,S)$, we first obtain an approximate answer by asking a query $\sel'_a(i,\tS)$. 
Let $i'=\sel_1(i,D_a)$ be the rank of the $i$-th symbol $a$ that is not marked as deleted. Let $l_0=\sel'_a(i',\tS)$. We find $l_1=\ra_1(l_0,\tE)$ and $l_2=\sel_a(\ra_a(l_1,\tS)+1,\tS)$. 
Let $first=\sel_1(l_1,\tE)$ and $last=\sel_1(l_2,\tE)$ be the positions of  $\tS[l_1]$ and $\tS[l_2]$ in $S$.
By definition of $\sel'$, $\ra_a(first,S)\le i$ and $\ra_a(last,S)> i$.  If $\ra_a(first,S)=i$, then obviously $\sel_a(i,S)=first$.
Otherwise the answer to $\sel_a(i,S)$ is an integer between $first$ and $last$.
By definition of $\tS$, the substring $S[first]$, $S[first+1]$, $\ldots$, $S[last]$ contains at most $r$ occurrences of $a$.
All these occurrences are stored in  subsequences $S_j$ for  $j\ge 1$.
We compute $i_0=\ra_a(\ra_0(first,R),S_0)$ and $i_1=i'-i_0$. We find the index $t$ 
such that $\sum_{j=1}^{t-1} Count_a[j]< i_1\le \sum_{j=1}^t Count_a[j]$.
Then $v_1=\sel_a(i_1- \sum_{j=1}^{t-1} Count_a[j], S_t)$ is the position of $S[\sel_a(i,S)]$ in $S_t$. We find its index in $S$ by computing  $v_2=v_1+\sum_{j=1}^{t-1}Size[j]$ and $v_3=\sel_1(v_2,R)$. Finally  $\sel_a(i,S)=\ra_1(v_3,M)$.

Answering an $\acc$ query is straightforward. We determine whether $S[i]$ is stored in $S_0$ or in some $S_j$ for $j\ge 1$ using $R$. Let $i'=\sel_1(i,M)$. If $R[i']=0$ and $S[i]$ is stored in $S_0$, then $S[i]=S_0[\ra_0(i',R)]$. If $R[i']=1$, we compute $i_1=\ra_1(i',R)$  and find the index $j$ such that 
$\sum_{t=1}^{j-1}Size[t] < i_1 \le \sum_{t=1}^j Size[t]$. The answer to $\acc(i,S)$ is $S[i]=S_j[i_2]$ for $i_2=i_1-\sum_{t=1}^{j-1}Size[t]$.

\paragraph{Space Usage.}
The redundancy of our data structure can be estimated as follows. The space needed to keep the symbols that are marked as deleted in subsequences $S_j$ is bounded by $O((n/r)(\log \sigma+\log r))$: Let $\on_a$ denote the number of symbols $a$ that are marked as deleted and let $\on=\sum_a\on_a$. Then all symbols that are marked as deleted use $X=\sum_a\on_a\log\frac{n+\on}{n_a+\on_a}$ bits. Since $\frac{n+\on}{n_a+\on_a}\le \frac{2n}{\on_a}$, $X\le \sum_a\on_a + \sum_a \on_a\log\frac{n}{\on_a}$. If $\on < n/r^2$, $X=o(n)$. If $\on> n/r^2$, then  $X=O(n/r)+ O(n \log r)+\sum_a\on_a\log \frac{\on}{\on_a}= O(\frac{n}{r}(\log\sigma+\log r))$.  $S_0$ also takes $O((n/r)\log \sigma)$ bits.
The bit sequences $R$ and $M$ need $O((n/r)\log r)=o(n)$ bits; $\tB$, $\tE$ also use $O((n/r)\log r)$ bits. 
 Each bit sequence $D_a$ can be maintained in $O(n'_a\log (n_a/n'_a))$ bits where $n_a$ is the total number of symbols $a$ in $S$ and $n'_a$ is the number of symbols $a$ that are marked as deleted. All $D_a$ take $O(\sum_{a\in \Sigma} n'_a \log\frac{n_a}{n'_a})$. To estimate the last expression, we divide the alphabet $\Sigma$ into $\Sigma_1$ and $\Sigma_2$;  $Sigma_1$ contains all symbols $a$ such that $n'_a\ge n_a/\log^2n$ and $\Sigma_2$ contains all symbols $a$, such that $n'_a< n_a/\log^2n$. Then $\sum_{a\in \Sigma}n'_a\log\frac{n_a}{n'_a}=\sum_{a\in \Sigma_1}n'_a\log\frac{n_a}{n'_a}+\sum_{a\in \Sigma_2}n'_a\log\frac{n_a}{n'_a}\le 
(2n/r)\log\log n +(n/\log n)=O((n/r)\log\log n)$. Hence all $D_a$ need 
$O((n/r)\log\log n)=o(n)$ bits.  
The subsequence $\tS$ can be stored in 
$O((n/r)\log \sigma)$ bits. Thus all auxiliary subsequences use 
$O((n/r)(\log\sigma + \log r))=O(n\frac{\log\sigma\log\log n}{\log n})$ bits. 
Data structures for subsequences $S_i$, $r\ge i\ge 1$, use  $\sum_{i=1}^r(n_iH_k(S_i)+o(n_i\log\sigma))=nH_k(S\setminus S_0)+o(n\log\sigma)$ bits for any $k=o(\log_{\sigma}n)$, 
where $n_i$ is the number of symbols in $S_i$.  
Since $H_k(S)\le H_0(S)$ for $k\ge 0$, all subsequences $S_i$ are stored in $nH_0(S)+o(n\log\sigma)$ bits. 


\paragraph{Updates.}
When a new symbol is inserted, we insert it into the subsequence $S_0$ and update the sequence $R$. 
The data structure for $\tS$ is also updated accordingly. We also insert a $1$-bit at the appropriate position of bit sequences $M$ and $D_a$ where $a$ is the inserted symbol. Deletions from $S_0$ are symmetric.  
When an element is deleted from $S_i$, $i\ge 1$, we replace the $1$-bit corresponding to this element in $M$  with a $0$-bit. We also  change the appropriate bit in $D_a$ to $0$, where $a$ is the symbol that was deleted from $S_i$.

We must guarantee that the number of elements in $S_0$ is bounded by $O(n/r)$; the number of elements marked as deleted must be also bounded by $O(n/r)$. Hence we must re-build the data structure when the number of symbols in $S_0$ or the number of deleted symbols is too big. Since we aim for updates with worst-case bounds, the cost of re-building is distributed among $O(n/r)$ updates. We run two processes in the background. The first background process moves elements of $S_0$ into subsequences $S_i$. The second process purges sequences $S_1$, $\ldots$, $S_r$ and removes all symbols marked as deleted from these sequences. Details are given in Section~\ref{sec:updatesbackground}.

We assumed in the description of updates that $\log n$ is fixed. In the general case we need additional background processes that increase or decrease sizes of subsequences when $n$ becomes too large or too small. These processes are organized in a standard way. \no{and will be described in the full version of this paper.}
Thus we obtain the following result
\begin{lemma}
  \label{lemma:midsigma}
A dynamic string $S[1,n]$  over alphabet $\Sigma=\{\,1,\ldots,\sigma\,\}$ for $\sigma<n/\log^3 n$
can be stored in a data structure that uses $nH_0 + O(n\frac{\log\sigma\log\log n}{\log n})+ O(n(\log\log \sigma)^3)$ bits and answers queries
$\acc$, $\ra$ and $\sel$ in time $O(\log n/\log \log n)$.
Insertions and deletions of symbols are supported in $O(\log n /\log \log n)$ time. 
\end{lemma}
\shortver{In the full version of this paper we show that the space usage of the above described data structure can be reduced to $nH_k + o(n\log\sigma)$ bits. We also show how the result of Lemma~\ref{lemma:midsigma} can be extended to the case when $\sigma\ge n/\log^3 n$.  The full version also contains the description of the static data structure and presents the procedure for extracting a substring $S[i..i+\ell]$ of $S$ in 
$O(\log n/\log \log n +\ell)$ time.}

\subsection{Compressed Data Structure for $\sigma > n/\log^3 n$}
\label{sec:bigsigma}
If the alphabet size $\sigma$ is almost linear, we cannot afford storing the arrays 
$Count_a[]$. 
Instead, we keep a bit sequence $BCount_a$ for each alphabet symbol $a$.
Let $s_{a,i}$ denote the number of $a$'s occurrences in the subsequence $S_i$ and $s_a=\sum_{i=1}^r s_{a,i}$. Then  $BCount_a=1^{s_{a,1}}01^{s_{a,2}}0\ldots 1^{s_{a,r}}$. 
If $s_a<r\log^2n$,we can keep $BCount_a$ in $O(s_a\log\frac{r + s_a}{s_a})= O(s_a\log\log n)$ bits.
If $s_a>r\log^2 n$, we can keep $BCount_a$ in $O(r\log\frac{r + s_a}{s_a})= O((s_a/\log^2n)\log n)=O(s_a/\log n)$ bits.
Using $BCount_a$, we can find for any $q$ the subsequence $S_j$, such that $Count_a[j]< q \le Count_a[j+1]$ in $O(\log n/\log \log n)$ time.

We also keep an effective alphabet\footnote{An alphabet for $S_j$ is effective if it contains only symbols that actually occurred in $S_j$.} for each $S_j$. We keep 
a bit vector $Map_j[]$ of size $\sigma$, such that $Map_j[a]=1$ if and only if $a$ occurs in $S_j$. Using $Map_j[]$, we can map a symbol $a\in [1,n]$ to a symbol $map_j(a)=\ra_1(a,Map_j)$ 
so that $map_j(a)\in [1,|S_j|]$ for any $a$ that occurs in $S_j$. Let $\Sigma_j=\{\, map_j(a)\,|\, a\text{ occurs in }S_j\,\}$. For every $map_j(a)$ we can find the corresponding symbol $a$ 
using a $\sel$ query on $Map_j$. We keep a static data structure for each sequence $S_j$ over $\Sigma_j$.
Queries and updates are supported in the same way as in Lemma~\ref{lemma:midsigma}. Combining the result of this sub-section and Lemma~\ref{lemma:midsigma}, we obtain 
the data structure for an arbitrary alphabet size.
\begin{theorem}
  \label{theor:anysigma}
A dynamic string $S[1,n]$  over alphabet $\Sigma=\{\,1,\ldots,\sigma\,\}$ 
can be stored in a data structure that uses $nH_0 + O(n\frac{\log\sigma\log\log n}{\log n})+ O(n(\log\log \sigma)^3)$ bits and answers queries $\acc$, $\ra$ and $\sel$ in time $O(\log n/\log \log n)$.
Insertions and deletions of symbols are supported in $O(\log n /\log \log n)$ time. 
\end{theorem}

\section{Compressed Data Structure II}
\label{sec:compr2}
By slightly modifying the data structure of Theorem~\ref{theor:anysigma} we can reduce the space usage to essentially $H_k(S)$ bit per symbol for any $k=o(\log_{\sigma}n)$  simultaneously. First, we observe that any sub-sequence $S_i$ for $i\ge 1$ is kept in a data structures that consumes $H_k(S_i)+o(|S_i|\log \sigma)$ bits of space. Thus all $S_i$ use $\sum_{i=1}^r(n_iH_k(S_i)+o(n_i\log\sigma))=nH_k(S\setminus S_0)+o(n\log\sigma)$ bits.  It can be shown that 
$nH_k(S\setminus S_0)=nH_k(S)+O(n (1+\frac{\log n}{r}))$ bits; we prove this bound in Section~\ref{sec:anal}.  Since $r=O(\log n/\log\log n)$, the data structure of Theorem~\ref{theor:anysigma} uses $nH_k + O(n) +O(n\log\log n)+O(n(\log\log \sigma)^3)$ bits. 

In order to get rid of the $O(n\log\log n)$ additive term, we use a different static data structure;  our static data structure is described in Section~\ref{sec:construct}. As before, the data structure for a sequence $S_i$ uses  $|S_i|H_k + o(|S_i|\log\sigma)$ bits.   But we also show in Section~\ref{sec:construct} that our static data structure can be constructed in $O(|S_i|/\log^{1/6}n )$ time if the alphabet size $\sigma$ is sufficiently small, $\sigma\le 2^{\log^{1/3}n}$.   The space usage $nH_k(S)+o(n\log\sigma)$ can be achieved by appropriate change of the parameter $r$.  
If $\sigma> 2^{\log^{1/3}n}$, we use the data structure of Theorem~\ref{theor:anysigma}.  As explained above, the space usage is $nH_k+o(n\log\sigma)+O(n\log\log n)=nH_k+o(n\log\sigma)$.
If $\sigma\le 2^{\log^{1/3}n}$ we  also use the data structure of Theorem~\ref{theor:anysigma}, but we set $r=O(\log n \log\log n)$ and implement static data structures as in  Section~\ref{sec:construct}. The data structure needs 
$nH_k(S)+O(n/\log\log n)+O(n(\log\log \sigma)^3)=nH_k(S)+o(n\log\sigma)$ bits. Since we can re-build a static data structure for a sequence $S_i$ in $O(|S_i|\log^{1/6}n)$ time, background processes incur an additional cost of $O(\log n/\log\log n)$. Hence the cost of updates does not increase.
 \begin{theorem}
  \label{theor:anysigma2}
A dynamic string $S[1,n]$  over alphabet $\Sigma=\{\,1,\ldots,\sigma\,\}$ 
can be stored in a data structure that uses $nH_k + O(n\frac{\log\sigma\log\log n}{\log n})+ O(n(\log\log \sigma)^3)$ bits and answers queries $\acc$, $\ra$ and $\sel$ in time $O(\log n/\log \log n)$.
Insertions and deletions of symbols are supported in $O(\log n /\log \log n)$ time. 
\end{theorem}

\section{Substring Extraction}
\label{sec:substr0}
Our representation of compressed sequences also enables us to retrieve a substring $S[i..i+\ell-1]$ of $S$.
The static data structure, described in Section~\ref{sec:construct} supports substring extraction in $O(\log n/\log\log n +\ell/\log_{\sigma} n)$ time. Hence we can quickly retrieve a substring of any $S_i$. We can also augment $S_0$ with an $O((n/r)\log \sigma)$ additional bits, so that a substring of $S_0$ is extracted in the same time.
We can retrieve a substring of $S$ by extracting a substring of $S_0$ and a substring of some $S_i$ for $i\ge 1$ and merging the result. A detailed description is provided in Section~\ref{sec:substr}.  Our  result can be summed up as follows.
\begin{theorem}
  \label{theor:substr}
We can augment  data structures described in Theorem~\ref{theor:anysigma} and Theorem~\ref{theor:anysigma2} with 
$O((n/r) \log\sigma)$ additional bits, so that a substring of length $ell$ can be extracted in $O((\log n/\log\log n)+ell/\log_{\sigma}n)$ time. The parameter $r=\Omega(\log n/\log\log n)$ is defined in the same way as in Theorems~\ref{theor:anysigma} and~\ref{theor:anysigma2}. 
\end{theorem}

\bibliographystyle{abbrv}
\bibliography{dynrank}

\begin{thebibliography}{10}

\bibitem{AlbersH97}
S.~Albers and T.~Hagerup.
\newblock Improved parallel integer sorting without concurrent writing.
\newblock {\em Inf. Comput.}, 136(1):25--51, 1997.

\bibitem{AnderssonHNR95}
A.~Andersson, T.~Hagerup, S.~Nilsson, and R.~Raman.
\newblock Sorting in linear time?
\newblock In {\em Proc. 27th Annual ACM Symposium on Theory of Computing (STOC
  1995)}, pages 427--436, 1995.

\bibitem{ArgeV03}
L.~Arge and J.~S. Vitter.
\newblock Optimal external memory interval management.
\newblock {\em {SIAM} J. Comput.}, 32(6):1488--1508, 2003.

\bibitem{BGNN10}
J.~Barbay, T.~Gagie, G.~Navarro, and Y.~Nekrich.
\newblock Alphabet partitioning for compressed rank/select and applications.
\newblock In {\em Proc. 21st ISAAC}, pages 315--326 (part II), 2010.

\bibitem{BHMR07}
J.~Barbay, M.~He, J.~I. Munro, and S.~S. Rao.
\newblock Succinct indexes for strings, binary relations and multi-labeled
  trees.
\newblock {\em ACM Transactions on Algorithms}, 7(4):article 52, 2011.

\bibitem{BN12}
D.~Belazzougui and G.~Navarro.
\newblock New lower and upper bounds for representing sequences.
\newblock In {\em Proc. 20th ESA}, LNCS 7501, pages 181--192, 2012.

\bibitem{BenderCDFZ02}
M.~A. Bender, R.~Cole, E.~D. Demaine, M.~Farach-Colton, and J.~Zito.
\newblock Two simplified algorithms for maintaining order in a list.
\newblock In {\em Proc. 10th Annual European Symposium on Algorithms (ESA
  2002)}, pages 152--164, 2002.

\bibitem{BB04}
D.~Blandford and G.~Blelloch.
\newblock Compact representations of ordered sets.
\newblock In {\em Proc. 15th SODA}, pages 11--19, 2004.

\bibitem{CHLS07}
H.~Chan, W.-K. Hon, T.-H. Lam, and K.~Sadakane.
\newblock Compressed indexes for dynamic text collections.
\newblock {\em ACM Transactions on Algorithms}, 3(2):article 21, 2007.

\bibitem{CHL04}
H.-L. Chan, W.-K. Hon, and T.-W. Lam.
\newblock Compressed index for a dynamic collection of texts.
\newblock In {\em Proc. 15th CPM}, LNCS 3109, pages 445--456, 2004.

\bibitem{DietzS87}
P.~F. Dietz and D.~D. Sleator.
\newblock Two algorithms for maintaining order in a list.
\newblock In {\em Proc. 19th Annual ACM Symposium on Theory of Computing (STOC
  1987)}, pages 365--372, 1987.

\bibitem{FMMN07}
P.~Ferragina, G.~Manzini, V.~M{\"a}kinen, and G.~Navarro.
\newblock Compressed representations of sequences and full-text indexes.
\newblock {\em ACM Transactions on Algorithms}, 3(2):article 20, 2007.

\bibitem{FerraginaV07}
P.~Ferragina and R.~Venturini.
\newblock A simple storage scheme for strings achieving entropy bounds.
\newblock {\em Theor. Comput. Sci.}, 372(1):115--121, 2007.

\bibitem{FS89}
M.~Fredman and M.~Saks.
\newblock The cell probe complexity of dynamic data structures.
\newblock In {\em Proc. 21st STOC}, pages 345--354, 1989.

\bibitem{GMR06}
A.~Golynski, J.~I. Munro, and S.~S. Rao.
\newblock Rank/select operations on large alphabets: a tool for text indexing.
\newblock In {\em Proc. 17th SODA}, pages 368--373, 2006.

\bibitem{GGV03}
R.~Grossi, A.~Gupta, and J.~S. Vitter.
\newblock High-order entropy-compressed text indexes.
\newblock In {\em Proc. 14th SODA}, pages 841--850, 2003.

\bibitem{GrossiRRV13}
R.~Grossi, R.~Raman, S.~R. Satti, and R.~Venturini.
\newblock Dynamic compressed strings with random access.
\newblock In {\em Proc. 40th International Colloquium on Automata, Languages,
  and Programming (ICALP 2013)}, pages 504--515, 2013.

\bibitem{GuibasMPR77}
L.~J. Guibas, E.~M. McCreight, M.~F. Plass, and J.~R. Roberts.
\newblock A new representation for linear lists.
\newblock In {\em Proceedings of the 9th Annual {ACM} Symposium on Theory of
  Computing}, pages 49--60, 1977.

\bibitem{GHSV07}
A.~Gupta, W.-K. Hon, R.~Shah, and J.~S. Vitter.
\newblock A framework for dynamizing succinct data structures.
\newblock In {\em Proc. 34th ICALP}, pages 521--532, 2007.

\bibitem{HM10}
M.~He and J.~I. Munro.
\newblock Succinct representations of dynamic strings.
\newblock In {\em Proc. 17th SPIRE}, pages 334--346, 2010.

\bibitem{HSS03}
W.-K. Hon, K.~Sadakane, and W.-K. Sung.
\newblock Succinct data structures for searchable partial sums.
\newblock In {\em Proc. 14th ISAAC}, pages 505--516, 2003.

\bibitem{HSS11}
W.-K. Hon, K.~Sadakane, and W.-K. Sung.
\newblock Succinct data structures for searchable partial sums with optimal
  worst-case performance.
\newblock {\em Theoretical Computer Science}, 412(39):5176--5186, 2011.

\bibitem{JanssonSS12}
J.~Jansson, K.~Sadakane, and W.-K. Sung.
\newblock {CRAM}: Compressed random access memory.
\newblock In {\em Proc. 39th International Colloquium on Automata, Languages,
  and Programming (ICALP 2012)}, pages 510--521, 2012.

\bibitem{Kopelowitz12}
T.~Kopelowitz.
\newblock On-line indexing for general alphabets via predecessor queries on
  subsets of an ordered list.
\newblock In {\em Proc. 53rd Annual IEEE Symposium on Foundations of Computer
  Science (FOCS 2012)}, pages 283--292, 2012.

\bibitem{LP07}
S.~Lee and K.~Park.
\newblock Dynamic rank-select structures with applications to run-length
  encoded texts.
\newblock In {\em Proc. 18th CPM}, LNCS 4580, pages 95--106, 2007.

\bibitem{LP09}
S.~Lee and K.~Park.
\newblock Dynamic rank/select structures with applications to run-length
  encoded texts.
\newblock {\em Theoretical Computer Science}, 410(43):4402--4413, 2009.

\bibitem{MN06}
V.~M{\"a}kinen and G.~Navarro.
\newblock Dynamic entropy-compressed sequences and full-text indexes.
\newblock In {\em Proc. 17th CPM}, LNCS 4009, pages 307--318, 2006.

\bibitem{MN08}
V.~M{\"a}kinen and G.~Navarro.
\newblock Dynamic entropy-compressed sequences and full-text indexes.
\newblock {\em ACM Transactions on Algorithms}, 4(3):article 32, 2008.

\bibitem{Manzini01}
G.~Manzini.
\newblock An analysis of the burrows-wheeler transform.
\newblock {\em J. ACM}, 48(3):407--430, 2001.

\bibitem{MehlhornN90}
K.~Mehlhorn and S.~N{\"{a}}her.
\newblock Dynamic fractional cascading.
\newblock {\em Algorithmica}, 5(2):215--241, 1990.

\bibitem{Mor03}
C.~Mortensen.
\newblock Fully-dynamic two dimensional orthogonal range and line segment
  intersection reporting in logarithmic time.
\newblock In {\em Proc. 14th SODA}, pages 618--627, 2003.

\bibitem{MunroRRR12}
J.~I. Munro, R.~Raman, V.~Raman, and S.~S. Rao.
\newblock Succinct representations of permutations and functions.
\newblock {\em Theor. Comput. Sci.}, 438:74--88, 2012.

\bibitem{NavarroN13}
G.~Navarro and Y.~Nekrich.
\newblock Optimal dynamic sequence representations.
\newblock In {\em Proc. 24th Annual ACM-SIAM Symposium on Discrete Algorithms
  (SODA 2013)}, pages 865--876, 2013.

\bibitem{NavarroN13a}
G.~Navarro and Y.~Nekrich.
\newblock Optimal dynamic sequence representations (full version).
\newblock {\em submitted for publication}, 2013.

\bibitem{NS14}
G.~Navarro and K.~Sadakane.
\newblock Fully functional static and dynamic succinct trees.
\newblock {\em {ACM} Transactions on Algorithms}, 10(3):16, 2014.

\bibitem{PatrascuD04}
M.~Patrascu and E.~D. Demaine.
\newblock Tight bounds for the partial-sums problem.
\newblock In {\em Proc.\ 15th Annual ACM-SIAM Symposium on Discrete Algorithms
  (SODA 2004)}, pages 20--29, 2004.

\bibitem{RRR07}
R.~Raman, V.~Raman, and S.~S. Rao.
\newblock Succinct indexable dictionaries with applications to encoding k-ary
  trees, prefix sums and multisets.
\newblock {\em ACM Transactions on Algorithms}, 3(4):article 8, 2007.

\bibitem{Ruzic08}
M.~Ruzic.
\newblock Constructing efficient dictionaries in close to sorting time.
\newblock In {\em Proc. 35th International Colloquium on Automata, Languages
  and Programming (ICALP 2008)}, pages 84--95, 2008.

\end{thebibliography}

\newpage
\appendix
\renewcommand\thesection{A.\arabic{section}}

\section{Colored Predecessor Queries}
\label{sec:markpred}
In this section we consider predecessor queries on a linked list, called  colored predecessor queries.  The result of this section  is used in the proof of Lemma~\ref{lemma:logsigma0}.
Suppose that each entry in an ordered list $L$ is \emph{colored} with a symbol $a\in \Sigma$ from an alphabet $\Sigma=\{\,1,\ldots,\sigma\,\}$. We will also sometimes say that 
an entry $e$ contains a symbol $a$. 
A colored predecessor query $(e_q,a)$ for an entry $e_q\in L$ and a symbol $a\in \Sigma$ asks for the rightmost entry $e\in L$ 
that is   colored with $a$ and precedes $e_q$. We consider the problem of answering colored predecessor queries on a dynamic list $L$. 
This problem was previously considered by Kopelowitz~\cite{Kopelowitz12} who described  a randomized $O(\log\log m +\log\log \sigma)$-time solution.
Mortensen~\cite{Mor03} described an $O(\log\log m)$ time solution for the case $\sigma=\log^cn$ and a  constant $c$. 
We present here a deterministic solution for an arbitrarily large alphabet. This result is also of independent interest.

We start by describing a data structure that uses  more than linear space. Then we will show how the space usage can be reduced to linear and how the update time can be decreased. 
\begin{lemma}
  \label{lemma:colpred0}
Let $L$ be a list with $m\le n$ entries. There exists an $O(m\log^2 m)$-bit data structure that answers colored predecessor queries on $L$ in $O(\log \log m(\log\log \sigma)^2)$ time and supports insertions and deletions in $O(\log m)$ time.
\end{lemma}
\begin{proof}
For a symbol $a$, let $L_a$ denote the sublist of $L$ that consists of entries containing $a$. Each entry that contains $a$ 
is augmented with a pointer to the next and the previous entries in $L_a$.  We also store an order maintenance data structure on $L$. This data structure can determine in $O(1)$ time whether $e_1$ precedes $e_2$ in $L$ for two arbitrary  entries $e_1\in L$ and $e_2\in L$ in a dynamic list $L$. We refer to~\cite{BenderCDFZ02,Kopelowitz12} for a description of such a data structure. 

We keep a balanced tree $T_L$ on $L$. For a node $u\in T_L$, the set $Col(u)$ consists of all symbols $a$ such that at least one leaf descendant of $u$ contains $a$. In every leaf of $T_L$, we keep pointers to all its ancestors.  
For every $a\in C(u)$, we also keep 
$a.\min(u)$  and $a.\max(u)$; $a.\min(u)$ (resp.\ $a.\max(u)$) points to the leftmost (rightmost) element of $L$ in the subtree of $u$ colored with $a$. 

Suppose that we want to find the rightmost entry $e_a\in L$ that contains a symbol $a$ and precedes an entry $e_q\in L$. 
We look for the lowest ancestor $u$ of (the leaf that contains) $e_q$ such that $a\in Dict(u)$. Using binary search on $\log n$ ancestors of $e_q$, we can find $u$ in $O(\log \log n (\log\log \sigma)^2)$ time.  If $e_q$ is in the right subtree of $u$, then 
$e_a=a.\max(u_l)$ where $u_l$ is the left child of $u$. If $e_q$ is in the left subtree of $u$, then we find $e'_a=a.\min(u_r)$ where $u_r$ is the right child of $u$. The entry $e'_a$ is the leftmost entry that follows $e_q$. Hence the entry $e_a$ is the first occurrence of $a$ in $L$ before $e'_a$. In other words, $e_a$ precedes $e'_a$ in $L_a$. 

When a new element $e$ is inserted into $L$, we insert it into some leaf $l_e$ of $T_L$ and a new entry into the corresponding list $L_a$. Insertion into $L_a$ requires that we find the rightmost entry $e_a$ that is colored with $a$ and precedes $e$ in the list. This takes $O(\log m/\log\log n)$ time as described above. Then we visit all ancestors of $l_e$ in $T_L$. If necessary, we add $a$ to  $C(v)$ in each visited node $v$.
 We keep the tree balanced, using the algorithms of weight-balanced B-tree~\cite{ArgeV03}.
The cost of maintaining $T_L$ so that its height remains $O(\log m)$ is $O(\log m)$ per insertion.  Deletions are symmetric. When an element $e$ is deleted, we remove it from the list $L_a$ and update $Dict(u)$ in at most one ancestor of $e$. Then we remove the leaf that contains $e$. The weight-balanced B-tree is 
not modified after the deletion of a leaf. But when a fraction of leaves is deleted, we construct a new tree $T_L$ and discard the old instance of $T_L$. The process of re-building $T_L$ can be run in the background so that the total worst-case cost of deleting $e$ is  $O(\log m)$.
\end{proof}

\begin{lemma}
\label{lemma:colpred}
Let $L$ be a list with $m\le n$ entries. There exists an $O(m\log m)$-bit data structure that answers colored predecessor queries on $L$ in $O((\log \log \sigma)^2\log\log m )$ time and supports insertions and deletions in $O((\log \log\sigma)^2\log\log m)$ time.
\end{lemma}
\begin{proof}
We divide every $L_a$ into $O(|L_a|/\log m)$ blocks so that every block contains $O(\log^2 m)$ consecutive entries of  $L_a$.  If $L_a$ consists of more than one block, then we maintain the list $L'_a$ that contains the first entry from every block of $L_a$.   The list $L_1$ contains all elements of $L'_a$ for all symbols $a$. We keep $L_1$ in the data structure of Lemma~\ref{lemma:colpred0}. 
For any symbol $a$, all elements of $L_a$ are also stored in a data structure $T_a$ that supports finger searches~\cite{GuibasMPR77}:  For any element $e_q\in L$ and a \emph{finger} $e'_a\in L_a$, $T_a$ can return the rightmost entry 
$e_a$ that is colored with $a$ and precedes $e_q$ in $O(\log d)$ 
comparisons, where $d$ is the number of entries between $e'_a$ and $e_a$ in $L_a$. Finally we also keep the list $L$ in the union-split-find data structure~ of Mehlhorn~\cite{MehlhornN90}. Using this data structure, we can find the first $e'\in L_1$ that precedes any $e\in L$ in $O(\log\log m)$ time. The data structure of Mehlhorn et al.~\cite{MehlhornN90} uses $O(m)$ words and supports updates in $O(\log\log m)$ time.

In order to find $e_a$ colored with symbol $a$ that precedes $e_q$, we find the first entry $e'\in L_1$ that precedes $e_q$. Then we identify the first entry $e_a'$ colored with $a$ that precedes $e'$.  There are $O(\log^2 m)$ entries of $L_a$ between $e'_a$ and $e_a$. When $e'_a$ is known, we can find $e_a$ in $O(\log \log m)$ time using finger search on $T_a$. The total query time is dominated by the search in $L_1$ and equals $O(\log\log m(\log\log \sigma)^2)$. 

When a new entry $e$ of color $a$ is inserted, we update $L$. 
Then we find the position of $e$ in $L_a$ and update $L_a$ and $T_a$. We can maintain the sizes of blocks in lists $L_a$ so that each block consists of $O(\log^2 m)$ entries and  there is one insertion into $L_1$ for $O(\log m)$ insertions into $L$; details will be given in the full version.  Thus the total cost of an insertion is $O((\log \log \sigma)^2\log\log m )$. Deletions are symmetric. 
\end{proof}

\section{Prefix Sum  Queries on a List}
\label{sec:dl}
In this section we describe a data structure on a list $L$ that is used in the proof of Lemma~\ref{lemma:logsigma0} in Section~\ref{sec:ranklog}.
\begin{lemma}
  \label{lemma:dl}
We can keep  a dynamic list $L$ in an $O(m \log m)$-bit data structure $D(L)$, where $m$ is the number of entries in $L$. $D(L)$ can find the $i$-th entry in $L$ for $1\le i \le m$ in 
$O(\log m/\log \log n)$ time. $D(L)$ can also compute the number of entries before a given element $e\in L$ in $O(\log m/\log\log n)$ time. Insertions and deletions are also supported in $O(\log m/\log \log n)$ time.  
\end{lemma}
\begin{proof}
$D(L)$ is implemented as a balanced tree with node degree $\Theta(\log^{\eps}n)$. In every internal node we keep a data structure $\Pref(u)$; $\Pref(u)$ contains the total number $n(u_i)$ of elements stored below every child $u_i$ of $u$.
$\Pref(u)$ supports prefix sum queries (i.e., computes $\sum_{i=1}^t n(u_i)$ for any $t$) and finds the largest 
$j$, such that $\sum_{i=1}^j n(u_i) \le q$ for any integer $q$. 
We implement $\Pref(u)$  as in Lemma 2.2 in~\cite{PatrascuD04}
so that both types of queries are supported in $O(1)$ time.   $\Pref(u)$ uses linear space (in the number of its elements) and can be updated in $O(1)$ time. $\Pref(u)$ needs a look-up table of size $o(n^{\eps})$. To find the $i$-th entry  in a list, we traverse the root-to-leaf path; in each visited node $u$ we find the child that contains the $i$-th entry using $\Pref(u)$. 
To find the number of entries preceding a given  entry $e$ in a list, we traverse  the leaf-to-root path $\pi$ that starts in the leaf containing $e$. In each visited node $u$ we answer a query to $\Pref(u)$: if the $j$-th child $u_j$ of $u$ is on $\pi$, then we  compute $s(u)=\sum_{i=1}^{j-1}n(u_i)$ using $\Pref(u)$. The total number of entries to the left of $e$ is the sum of $s(u)$ for all  nodes $u$ on $\pi$. 
Since we spend $O(1)$ time in each visited node, both types of queries are answered in $O(1)$ time. 
An update operation leads to $O(\log m/\log \log n)$ updates of
data structures $\Pref(u)$. The tree can be re-balanced using the weight-balanced B-tree~\cite{ArgeV03}, so that its height is always bounded by $O(\log m/\log\log n)$.
\end{proof}

\section{Updating Data Structure in Lemma~\ref{lemma:logsigma}}
\label{sec:logsigmaupd}
When the size of a chunk $C_i$ equals $2\sigma$ we start the procedure of re-building this chunk. During the next $\sigma/2$ updates of $C_i$ we retrieve all elements of $C_i$ and insert them  into data structures for new chunks, $C'_i$ and $C''_i$. If an update is a deletion of some element $e$ and $e$ was already copied into $C_i'$ or $C_i''$, then we remove the copy of $e$
from $C_i'$ or $C_i''$. When all elements of $C_i$ are copied into $C_i'$ and $C_i''$, we say that a chunk $C_i$ is a copied
chunk. We keep ids of all copied chunks in a data structure $L_d$. Whenever a copied chunk $C_i$ is updated we also execute  the same update of $C_i'$ or $C_i''$.

We also run the following iterative procedure that replaces copied chunks with two chunks. Each iteration starts by finding a chunk $C_i$ with the largest number of elements. Then all arrays $B_a$ are updated in increasing order of $a$. We insert a $0$-bit at an appropriate position of $B_a$ so that $B_a=1^{d_1}0\ldots1^{d_i}0\ldots$ is changed to $B_a=1^{d_1}0\ldots1^{d'_i}01^{d''_i}0\ldots$ where $d_i$, $d_i'$ and $d_i''$ denote the number of $a$'s that occur in $C_i$, $C_i'$ and $C''_i$ respectively. We keep a variable $lastsym$ that equals the largest symbol $a$, such that $B_a$ is already updated. When all $B_a$ are modified in the above manner, we also update $B_t$ and change it from $B_t=1^{n_1}0\ldots1^{n_i}0\ldots$ to $B_t=1^{n_1}0\ldots1^{n'_i}01^{n''_i}0\ldots$ where $n_i$, $n_i'$ and $n_i''$ denote the total number of symbols in $C_i$, $C_i'$ and $C''_i$ respectively. Finally we delete the id of $C_i$ from 
$L_d$ set $lastsym=0$ and start the next iteration.  Every iteration takes $O(\sigma)$ time.  When a chunk is added to $L_d$, its size does not exceed $5\sigma/2$. Using Theorem 5 in~\cite{DietzS87}, we can show that the size of each chunk in $L_d$ grows by at most by $\sigma\cdot O(h_n)$ where $h_n=O(\log n)$ denotes the $n$-th harmonic number. 

We slightly modify the method for answering a $\sel$ query. 
Let $k$ denote the index of the last chunk that was retrieved from $L_d$. That is, the above described iterative procedure is currently changing bit vectors $B_a$ and $B_t$ changing 
$B_a=\ldots1^{d_k}0\ldots$ to $\ldots1^{d'_k}01^{d''_k}0\ldots$ and 
$B_t=\ldots1^{n_k}0\ldots$ to $\ldots1^{n'_k}01^{n''_k}0\ldots$.
To answer a query $\sel_a(i,S)$, we first find the index $i_2$
of the chunk $C_{i_2}$ that contains the $i$-th occurrence of $i$,
$i_2=\ra_0(\sel_1(i,B_a),B_a)+1$. 
If $i_2<k$ or $a>lastsym$, we proceed as described in the proof of Lemma~\ref{lemma:logsigma}.  
If $i\ge k$ and $a\le lastsym$, we decrement $i_2$ by $1$, $i_2=i_2-1$ and also proceed as in Lemma~\ref{lemma:logsigma}.

We also keep track of the number of chunks that contain no more than  $\sigma$ elements. If there are at least $n/2\sigma$ chunks containing at most $\sigma$ symbols, then we start a global re-building procedure. We retrieve all elements of $S$ and insert them into a new data structure. In the new data structure all elements are distributed among chunks, so that each chunk contains $\sigma$ elements. The global re-building process is executed during $n/4\sigma$ updates.

\section{Re-Building  Compressed Data Structure in the Background}
\label{sec:updatesbackground}
As shown in Section~\ref{sec:compr}, we must bound the total number of symbols in $S_0$ by $O(n/r)$ for a parameter $r$. 
We must also bound the number of symbols in $S_i$ for $i\ge 1$
that are marked as deleted by $O(n/r)$. 
We run two alternating processes in the background to satisfy these requirements.
 In order to bound the workspace we process sub-sequences $S_i$ one-by-one. For every $i$, $1\le i\le r$, we produce a new version $S'$ of $S_i$ containing  all relevant  elements of $S_0$ (i.e., all elements of $S_0$ that precede the first element of 
$S_{i+1}$ and follow the last element of $S_{i-1}$ in $S$). 
In order to navigate in the new version of $S_i$, we must modify parts of auxiliary sequences (such as $R$, $\tS$, $\tE$, and $\tB$). Therefore our background process also produces new versions for the relevant portions of auxiliary sequences. When the new version of $S_i$ is created, we discard the old version; we also replace the parts of auxiliary sequences with their new versions.  The second background process removes elements marked as deleted and updates $S_i$ in the same manner. A more detailed description follows. 

We conceptually divide $S_0$ into $r$ substrings $S_{0,i}$ for $1\le i \le r$. An element $e\in S_0$ is in $S_{0,i}$ for $1< i < r$ iff $e$ precedes the first element of $S_{i+1}$ in $S$  and follows the last element of $S_{i-1}$ in $S$. An element $e\in S_0$ is in $S_{0,1}$ if $e$ precedes the first element of $S_2$; $e\in S_0$ is in $S_{0,r}$ if $e$ follows the last element of $S_{r-1}$. 
Likewise the sequence $\tS$ is conceptually divided into $r$ substrings $\tS_1,\ldots,\tS_r$.
An element $e\in \tS$ is in $\tS_i$ for some $i\ge 1$ if $e$ is a copy of some $e'\in S_i$ or $e$ is a copy of some $e'\in S_0$ and $e'\in S_{0,i}$. 
We conceptually divide the binary sequence $R$  using the same principle:  
$R[j]$ is in $R_i$ if the $j$-th element of $S$ is from $S_i$ or the $j$-th element of $S$ is some $e'\in S_0$ such that  $e'\in S_{0,i}$. Other binary sequences are divided in the same way.  The procedure for moving elements of $S_0$ into $S_i$ for some $i$, $1\le i\le r$, is as follows.
\begin{description}
\item[Step 1]
We start by creating a new instance $S^c$ of $S_i$ and a new instance $\tS^c$ of $\tS^i$; we also create new instances of 
$R^i$ and the $i$-th parts of other binary sequences; namely $R^c$, $D^c_a$ for all $a\in \Sigma$ such that $a$ occurs in $S_i$, $\tB^c$ and $\tE^c$ are copies of $R_i$, $D_{a,i}$, $\tB_i$ and $\tE_i$ respectively. The cost of creating new instances for parts of auxiliary sequences can be distributed 
among the following updates of $S$, as will be explained below. At the end of Step 1, $R^c$ is a copy of $R_i$; likewise  $D^c_a$, $\tB^c$ and $\tE^c$ are copies of  $D_{a,i}$, $\tB_i$ and $\tE_i$ respectively. These newly created sequences will be called copy sequences.\\
\item[Step 2]
Then we insert the elements of $S_{0,i}$ at appropriate positions of $S^c_i$.  
We  modify the sequence $\tS^c$  accordingly. 
Changes in $\tS^c$ and $S_i^c$ also lead to changes in copy sequences $R^c$,  $D^c_a$, $\tB^c$ and $\tE^c$. We distribute the cost of Step 2 among  updates of $S$.  We will say that all  elements that are kept in $S_0$ (resp. in $S^c_0$) upon completion of Step 1  are \emph{old} elements.
 When a sequence $S$ is updated, we spend $O(\log n/\log \log n)$ time on the following actions: (i) we find the next unprocessed element $e_n$ in $S_0^c$ (symbols in $S^c_0$ are processed in the left-to-right order);  we set the bit corresponding to $e_n$ in $R^c$ to $1$ (ii) we insert $e_n$ at appropriate position of    $S_i^c$ (iii) if necessary, we update $\tS^c$; copy sequences $\tB^c$ and $\tE^c$ are updated accordingly.
We may also need to  update copy sequences after an update of $S$. 
If the update of $S$ is an insertion, and a new element $e$ is inserted into $S_{0,i}$, then we also insert $e$ into $S_0^c$. If an element $e$ is deleted and $e\in S_{0,i}$, then we remove the copy of $e$ from $S^c_0$; changes in $S_0^c$ can also lead to changes in $\tS^c$.  If a symbol $a$ is deleted from $S_i$, then we update $D^c_a$  accordingly. 
\item[Step 3]
When $S_i^c$ is completed, we discard old $S_i$, set $S_i=S^c_i$, and start using the new $S_i$ from now on. Simultaneously we replace the relevant section of 
$\tS$ with $\tS^c$. We also replace the relevant parts of $R$,  $D_a$,  $\tB$ and $\tE$ with $R^c$, $D^c_a$,  $\tB^c$ and $\tE^c$.
\end{description}
In order to execute the above background process, we must implement binary sequences, so that two additional procedures are supported: A binary sequence of length $m$ is divided into $r$ sectors (substrings) of length $O(m/r)$ each. We can produce a copy of each sector. The cost of producing a copy is distributed among $\frac{m\log\log n}{2r \log n}$ updates; when the procedure is finished, the sector and its copy are equal. We can perform updates on the original sequence and on a sector copy. We can also replace a sector with its copy and discard the original sector. 
Same procedures are also supported for the non-binary sequence $\tS$.
We can implement these procedures in such way that the cost of $\ra$, $\sel$, $\acc$, and updates is not increased. Implementation of auxiliary procedures is explained in \shlongver{the full version of this paper, attached at the end of this submission}{Section~\ref{sec:seqaddition}}.

Step 1 of the above process takes $O(n/r)$ time. Step 2 (insertion of new elements into $S_i$) takes $O(v_i(\log n/\log\log n))$ time, where $v_i$ is the number of elements inserted into $S_i'$. Step 3 takes $O(\log n/\log\log n)$ time. Thus old elements of $S_0$ are moved to $S_i$ for $i\ge 1$ in 
$O(n)+\sum_i v_i(\log n/\log\log n) =O(n)$ time. This process can be distributed among $n/4r$ updates.

The process of purging the sequences $S_1$, $\ldots$, $S_r$ is based on the same approach. 
For each $i=1,\ldots ,r$, we create a new instance of $S_i$ without  deleted elements; then we discard the old instance 
and start using the new version of $S_i$. Relevant parts of $\tS$ and binary sequences are also updated.
The re-building of $S_i$ is implemented in the same way as in the procedure of moving elements from $S_0$ to $S_i$ for $i\ge 1$.  The cost of purging $S_i$ is distributed among $n/4r$ following updates. 
Two above described background processes are run alternatingly; the first process starts when the either the number of elements in  $S_0$ or the number of elements marked as deleted is equal to $n/4r$.  In this way we guarantee that the number of elements in $S_0$ and the number of deleted elements does not exceed $n/r$.

\section{Auxiliary Procedures for Binary and Non-Binary Sequences}
\label{sec:seqaddition}
In this section we show how a sequence $S$ can be stored  in such a way that additional processes that create a copy for a part 
of $S$ are supported. Furthermore we can update the copied part and later replace the original part with its modified copy. 
We start by describing a binary sequence that supports an additional operation $\init(S,m)$; $\init(S,m)$ initializes an empty sequence of length $m$ that consists of $m$ $0$-bits.  Recall that $\lambda=\log n/\log\log n$.
\begin{lemma}
\label{lemma:ints}
  A binary sequence $S$ that supports $\ra_1(i,S)$, $\sel_1(i,S)$, $\acc(i,S)$, insertions, deletions and $\init(S,m)$ for any $m\le n$  can be stored in $O(s\log\frac{n}{s})+o(n)$ bits, where $s$ is the number of $1$-bits and $n$ is the length of the sequence. All operations, except for $init(S,m)$ take $O(\lambda)$ time; $init(S,m)$ can be executed in $O(1)$ time.
\end{lemma}
\begin{proof}
We divide the sequence $S$  into blocks $B_i$ such that each $B_i$ consists of $\Theta(\log^2 n)$ bits.  Each block is further divided into sub-blocks of $\Theta(\log^{1/2}n)$ bits. We will say that a block or a sub-block is non-empty if it contains at least one $1$-bit. A doubly-linked list $L$ contains one entry for each non-empty block.  We also keep a list $L_i$ for every block $B_i$ that contains $1$-bits; $L_i$ contains one entry for each non-empty sub-block.  For each entry $e_i$ of $L$ we keep the number of $1$'s in the corresponding block $B_i$; we also keep the total  number of bits in blocks $B_{j+1}$, $B_{j+2}$, $\ldots$, $B_{i-1}$, where $B_j$ is the rightmost  non-empty block that precedes $B_i$. We maintain a data structure that enables us to find the block that contains the $i$-th bit in the sequence. We also maintain a data structure that can find the block containing the $i$-th $1$-bit (or $0$-bit) and the number of $1$-bits ($0$-bits) that precede a specified block. We maintain the same data structure for each sub-block. All these data structures are implemented as balanced trees with node degree $\log^{\eps}n$ for a small constant $\eps>0$. Each node is augmented with additional information about the number of $1$-bits (resp.\ the total number of bits) in the subtrees of its children. Implementation is the same as for data structures $D(L)$ and $D(L_a)$ in Lemma~\ref{lemma:logsigma0}.

 Positions of $1$-bits in the same sub-block are difference coded: for every $1$-bit we store the difference 
between its position and the position of the preceding $1$-bit 
in the same block; for the first $1$-bit in the block, we store 
its position in the block.  
The list $L$ and its data structures can be kept in $O(n/\log n)$ bits. All lists $L_i$ and their data structures are kept in $O(n(\log\log n/\sqrt{\log n}))$ bits. Difference coding of $1$-bits in all blocks consumes $O(s\log\frac{n}{s})$ bits.

To answer a query $\ra_1(i,S)$ we find the block $B_i$ and its sub-block $B_{i,j}$ containing the $i$-th bit. 
Then we find the number of $1$-bits that precede $B_k$ in $L$ and the number of $1$-bits that precede $B_{k,j}$ in $L_k$. 
We can find the number of $1$-bits that precede the bit with global position $i$ in $B_{k,j}$ using a look-up table.  
Summing three above values, we obtain $\ra_1(i,S)$.  Queries $\sel_1(i,S)$, $\ra_0(i,S)$, and $\ra_1(i,S)$ are computed in a similar way. Thus all queries are answered in $O(\log n/\log\log n)$ time.

Since we only keep non-empty blocks and sub-blocks, operation $\init(S,m)$ takes constant time.
Insertions and deletions are implemented as in previously known  data structures supporting rank and select on binary sequences. When an element is inserted, we find its block $B_i$  and its sub-block; we insert the new element into its sub-block and  update lists $L$ and $L_i$ if necessary. We maintain sizes of blocks and sub-blocks using standard techniques. Deletions are symmetric. 
Hence insertions and deletions are supported in $O(\log n/\log \log n)$ time. 
\end{proof}

Now we describe how a copy of a binary sequence $S$ can be created. Let $\lambda=\log n/\log\log n$.  
\begin{lemma}
\label{lemma:bincopy}
  Let $S$ be a binary sequence of length $s$. Procedure $copy()$, that produces a copy of $S$, can be implemented as a background process that runs during $O(s/\log n)$ consecutive updates.
We can support updates on the original sequence and its copy in $O(\lambda)$ time. Operations $\ra$, $\sel$, and $\acc$ are 
executed in $O(\lambda)$ time.  The underlying data structure uses $sH_0(S)+s+o(s)$ bits. 
\end{lemma}
\begin{proof}
The procedure for creating a copy $S'$ of $S$ consists of two stages. During the first stage we produce a copy of $S$. $S$ is represented in the same way as in~\cite{NS14}. As described in~\cite{NS14}, $S$ is split into chunks 
and we maintain data structures that  support counting the number of  $0$-bits (resp.\ $1$-bits) among the chunks and searching for the chunk that contains the $i$-th $0$-bit (or $1$-bit). We can create a copy by copying the original sequence of chunks. The data structure that supports counting and searching among chunks is essentially a tree with $O(s')$ nodes; we can create this tree in $O(s')$ time, where $s'=O(s/\log n)$ is the number of chunks. 

Thus the background process that creates a copy of $S$ takes $O(s/\log n)$ time. We can distribute its cost among $O(s/(\lambda\log n))$ updates where $\lambda=\log n/\log\log n$. We keep information about these updates in four data structures. The data structure $U$ keeps information about positions of updates: the $i$-th $1$-bit in a sequence
$U$ is the position of the $i$-th update (insertion or deletion) in $S$.
Thus $U$ contains one bit 
for every element in $S$ and one bit for every  element that was deleted from $S$. Updates are counted in the left-to-right order and $U$ is implemented as in Lemma~\ref{lemma:ints}. We also keep a bit sequence $T$ which indicates
the type of updates on $S$: $T[i]=0$ if the $i$-th update stored in $U$ is 
a deletion and $T[i]=1$ if the $i$-th update is an insertion. The sequence $B_n$ 
contains the values of elements inserted into $S$. A sequence $U_d$ helps us navigate between $S$ and $U$; $U_d$ contains one bit 
for every element in $S$ and one bit for every  element that was deleted from $S$. If $U_d[j]=1$, then the corresponding element was already deleted from $S$; if $U_d[j]=0$, then the corresponding element is in $S$.
During each update, we perform the following operations: 
\begin{itemize*}
\item 
a new element is inserted into or deleted from $S$ at position $i$\\
\item Let $i'=\sel_0(i,U_d)$ be the position of the $i$-th $0$-bit in $U_d$. If the update of $S$ is an insertion, we insert a $1$-bit into $U_d$ at position $i'$. If the update is a deletion, we replace the $i'$-th bit in $S$ with a $0$-bit (replacement is implemented by deleting a $1$-bit and inserting a $0$-bit at the same position). \\
\item If the update of $S$ is an insertion, we insert a $1$-bit at position $i'$ into $U$; if the update is a deletion we replace the $i'$-th bit of $U$ with a $1$-bit. We also insert a  bit indicating the type of update into the sequence $T$. If an update is an insertion, we add the value of a new bit  into a bit sequence $B_n$.  \\
\item  we spend $O(\lambda)$ time on constructing a copy sequence.
\end{itemize*}

The first stage is finished after $O(s/(\lambda\log n))$ updates of $S$. When the first stage is completed, $S$ and its copy sequence $S'$ differ because $O(s/(\lambda\log n))$ most recent updates changed the original sequence but were not performed on its copy. During the second stage we synchronize $S$ and $S'$. The synchronisation procedure is also distributed among $O(s/\lambda\log n)$ updates. During every update operation, we proceed as follows:\\
- a new element is inserted into $S$ or deleted from $S$. We also change the copy sequence $S'$ accordingly. If the position $i$ of an element in $S$ is known, then we can find its position $i_n$ in $S'$ using sequences $U_d$, $U$ and $T$.
Using $U_d$, we find the position $i_d=\sel_0(i,U_d)$ corresponding to $i$ in $U$.
Using $U$, we find the number $u$ of updates that precede $i_d$; using $T$ we can find the number of insertions and deletions among the first $u$ updates.\\ 
- we also execute updates stored in sequences $U$, $T$, and $B_n$. We retrieve 
the position $i=\sel_1(1,U)$ of the first $1$-bit stored in $U$ and find the position $i'$ in $S'$ that corresponds to the position $i$ in $S$. Then we either insert a new element at the position $i'$ or remove the $i'$-th element from $S'$ according to the data stored in $T$ and $B_n$. Finally we delete the $i$-th bit from $U$ and $U_d$. We also delete the corresponding bit from $T$ and remove the corresponding symbol from $B_n$ (if the processed update is an insertion). \\
At the end of the second stage $S'$ and $S$ are equivalent.
\end{proof}

Now we consider the sequence $S$ of length $s$ that is divided into $s/r$ contiguous parts for $r=\log^{O(1)}n$. Each part, called a sector of $S$, consists of $O(s/r)$ elements. The procedure $copysector()$ creates a copy of an arbitrary sector. The procedure $copysector()$ can be executed in the background during a sequence of $s/(r\lambda)$ updates. Furthermore we can split a sector into two sectors and merge two adjacent sectors in the same time. Last, we can also replace a sector with its copy in $O(r)$ time. Update operations are supported on both $S$ itself and on the copy of a sector (we assume that at any time a copy of only one sector is created or used). 
\begin{lemma}
  \label{lemma:bincopysec}
Let $S$ be a binary sequence of length $s\le n$  and let $S$ be divided into $r$ sectors of $O(s/r)$ symbols. Procedure $copysector()$, that produces a copy of a sector, can be implemented as a background process that runs during $O(s/( r\lambda \log n))$ consecutive updates. Procedures $splitsector()$ and $mergesectors()$ can be executed in the same way.
Operation $replacesector()$ can be executed in $O(\log\log n)$ time.  The underlying data structure uses $sH_0(S)+o(n)$ bits. 
\end{lemma}
\begin{proof}
  Every sector is maintained in the data structure of Lemma~\ref{lemma:bincopy}. Furthermore we maintain a sequence $G$ that keeps the numbers of elements in every sector. Sequences $G_0$ and $G_1$ maintain the number of $0$'s and $1$'s in every sector. Using $G$ and data structures for individual sectors, we can answer rank, select, and access queries on $S$. Procedure $copysector()$ is implemented as $copy()$ on a sector of $S$. When a copy of a sector is ready, we can support updates on this copy. Besides we can also replace a sector with its copy and update the data structure on the sequence $G$ accordingly; this operation takes $O(\log \log n)$ time. Splitting and merging of sectors is implemented in a similar way. Suppose that we want to split a sector $S_i$ into $S_i'$ and $S''_i$. We employ the same two-stage procedure that was used to create a copy of a sector. During the first stage we assign elements of $S_i$ to $S'_i$ and $S''_i$. Then we create the data structures for $S'_i$ and $S''_i$. Updates that are relevant for new versions are deposited in data structures $U_1$, $T_1$, $B_1$ and $U_2$, $T_2$, $B_2$ respectively. During the second stage we execute updates stored in $U_i$, $T_i$, and $B_i$ for $i=1,2$. Auxiliary data structures are realized in the same way as in the procedure $copysector()$. When new sectors $S_i'$ and $S_i''$ are ready, we replace $S_i$ with $S_i'$ and $S_i''$ and update $G$.   
\end{proof}

We can implement similar procedures for a sequence over a general alphabet. We assume, however, that copies of sectors are produced consecutively: first copy of the first sector is created; when the first sector is replaced with a (possibly modified) copy sector, we create the copy of the next sector, etc. In this scenario, it is easy to maintain the dynamic sequence that supports the copying, splitting and merging of sectors.
\begin{lemma}
\label{lemma:copygen}
  Let $S$ be a sequence of length $s\le n$ over an alphabet $\Sigma=\{\,1,\ldots\sigma\,\}$ and let $S$ be divided into $r$ sectors. Procedure $copynextsector()$, that produces a copy of a sector, can be implemented as a background process that runs during $O(s/(r\lambda))$ consecutive updates. Procedures $splitnextsector()$ and $mergesectors()$ can be executed in the same way.
Operation $replacenextsector$ can be executed in $O(\lambda)$ time. The data structure for $S$ uses $O(s\log\sigma)$ bits.  
\end{lemma}
\begin{proof}
   Elements of $S$ are distributed among two  dynamic data structure, $S_{old}$ and $S_{new}$. Both of them are implemented as in Lemma~\ref{lemma:logsigma}. Originally $S_{new}$ is empty and all elements of $S$ are in $S_{old}$. Procedure $copynextsector()$ traverses elements of the next sector and appends them at the end of the new sequence $S_{new}$. When $replacenextsector()$ is executed for the last (rightmost) sector, we set $S_{old}=S_{new}$ and $S_{new}=\emptyset$. Let $i_p$ denote the number of elements in all sectors of $S_{old}$ for which operation $replacenextsector()$ was executed. Let $i_m$ denote the total number of elements currently kept in $S_{new}$. We can answer $access(i,S)$ by retrieving $S_{new}[i]$ if $i\le i_m$ or retrieving $S[i-i_m+i_p]$ if $i>i_m$. We can answer a query $\ra_a(i,S)$ as follows. If $i\le i_m$, $\ra_a(i,S)=\ra_a(i,S_{new})$. If $i>i_m$, $\ra_a(i,S)=\ra_a(i-i_m+i_p,S_{old})$. 
To answer a query $\sel_a(i,S)$ we check whether $\ra_a(i_m,S_{new})\ge i$. If this condition is satisfied, then $\sel_a(i,S)=\sel_a(i,S_{new})$; otherwise $\sel_a(i,S)=\sel_a(i',S_{old})$ where $i'=i-\ra_a(i_m,S_{new})+\ra_a(i_p,S_{old})$.
\end{proof}

\section{Analysis}
\label{sec:anal}
We show that deleting $n/r$ symbols from a sequence $S$ does not increase too much the $k$-th order entropy. This result is needed in Section~\ref{sec:compr2} to prove the  space bound 
of $nH_k+o(n\log \sigma)+O(n (\log n)/r)$. 
Let $S=S[1]\ldots S[n]$. Let $S_0$ denote the subsequence of symbols that are deleted from $S$ and let $S_n=S\setminus S_0$. Let $|S_0|=n/r$ for a parameter $r$. We want to estimate  $|S_n|H_k(S_n) - |S|H_k(S)$ for some parameter $k\le \alpha\log_{\sigma}n-1$ and $\alpha<1$.

A context  $c_i$ is an arbitrary sequence of length $k$ over an alphabet $\sigma$; let $f_{a,i}$ denote the number of times 
a symbol $a$ is preceded by context $c_i$ in $S$ and $n_i=\sum_{a\in\Sigma} f_{a,i}$. The $k$-th order empirical entropy is defined as $\sum_{c_i\in \Sigma^k}\sum_{a\in\Sigma} f_{a,i}\log\frac{n_i}{f_{a,i}}$.

For a context $c_l$, let $n_l$ be the number of times it occurs in $S$ and let $n_l'$ be the number of times it occurs in $S\setminus S_0$.
Suppose that a symbol $S[i]$ is deleted. It changes the context for the next $k$ symbols $S[i+1],..., S[i+k]$. We will say that one deletion spoils $k$ symbols and moves them to a different context. 
If a symbol $S[l]$ is spoiled and the context of $S[l]$ in $S_n$ is $c_l$, then 
$S[l]$ is encoded with at most $\log (n'_l)$ bits. 
Let $p_l=n'_l-n_l$ be the number of new symbols in the context $c_l$. Let $f'_{a,l}$ be the frequency of a new symbol $a$ in $c_l$ (that is, the number of times a spoiled symbol 
$a$ appears in the context $c_l$ in $S_n$). Then the total 
encoding length of spoiled symbols in the context $c_l$ does not exceed  $\sum_a f'_{a,l}\log\frac{n'_l}{f'_{a,l}}$ where 
$\sum_a f'_{a,l}=p_l$.
By Jensen's inequality, $\sum_a f'_{a,l}\log\frac{n'_l}{f'_{a,l}}\le p_l \log \frac{n'_l}{p_l/\sigma}$. Summing over all contexts $c_l$, the total encoding length of spoiled symbols can be bounded by $\sum_l p_l (\log \frac{n'_l}{p_l}+\log\sigma)$.

The total number $m$ of symbols that are spoiled is between $n/r$ and $(n/r)k-1$ because each deletion spoils between $1$ and $k$ following symbols. The number of spoiled symbols does not exceed $n$ independently of $r$ and $k$.
Hence $\sum_l p_l\le (n/r)\log_{\sigma}n$. Besides 
$\sum_l p_l\le \sum_l n'_l \le n$. 
Therefore $\sum_l p_l\log\frac{n'_l}{p_l}=O(n)$. To prove the latter fact, we divide all contexts $c_l$ into classes $L_1$, $L_2$, $\ldots$, $L_{\log^*n}$. 
$L_i$ contains all context indices $l$, such that $f_{i-1}(n)>\frac{n'_l}{p_l}\ge f_i(n)$, where $f_0(n)=n$ and $f_i(n)=(\log^{(i)}n)^2$ for $i\ge 1$.
For any $L_i$, $\sum_{l\in L_i}p_l\log\frac{n'_l}{p_l}\le \frac{n}{(\log^{(i)}n)^2}\log (f_{i-1}(n))=\frac{n}{(\log^{(i)}n)^2}O(\log^{(i)}n)=
O(\frac{n}{\log^{(i)}n})$. Hence $\sum p_l\log\frac{n'_l}{p_l}=\sum_{i=1}^{\log^*n}\sum_{l\in L_i}p_l\log\frac{n'_l}{p_l}= n\sum_{i=1}^{\log^*n}O(\frac{1}{\log^{(i)}n})$. Hence $\sum\frac{n'_l}{p_l}=O(n)$ because 
$\sum_{i=1}^{\log^*n}\frac{1}{\log^{(i)}n}=O(1)$.
Thus the total encoding length of all spoiled symbols is bounded by $E_{add}=O(n(1+\frac{\log n}{r}))$. 

Another factor that may increase the encoding length  is that spoiled symbols are moved to new contexts and thus the encoding length of all other symbols in these new contexts slightly increases. Consider a context $c_i$  that occurred $n_i$ times in $S$ and $n_i+f_i$ times in $S\setminus S_0$ for some $f_i>0$. 
We say that a symbol $S[l]=a$ that follows $c_i$ in $S\setminus S_0$ is an old occurrence if this occurrence of $a$ is also preceded by $c_i$ in $S$. 
The  encoding length for all old occurrences in $S_n$ is $\sum_{a=1}^{\sigma} \occ_{a,i}\log\frac{n_i+f_i}{\occ_{a,i}}$. The total encoding length for the same occurrences in $S$ is  $\sum_{a=1}^{\sigma} \occ_{a,i}\log\frac{n_i}{\occ_{a,i}}$.
The difference between encoding lengths of old occurrences in $S$ and $S_n$ 
is $inc(c_i)=n_i\log \frac{n_i+f_i}{n_i}$. If $f_i\le n_i$, then $inc(c_i)\le n_i$. Summing up the differences over all  contexts $c_i$ such that 
$f_i\le n_i$, we obtain $E_1\le \sum_i n_i\le n$. 
If $f_i> n_i$, then $inc(c_i)\le n_i(\log\frac{f_i}{n_i} +1)$
Summing up over all contexts $c_i$ such that $f_i>n_i$, we get  $E_2=\sum_{i} n_i(\log \frac{f_i}{n_i}+1)$. Since $\sum_{i} n_i < \sum_{i} f_i \le n$, 
$\sum_i n_i < n$ and $\sum_i n_i\log\frac{f_i}{n_i} \le n$.
Hence, $E_2=O(n)$.  

Thus $|S_n|H_k(S_n) - |S|H_k(S)=E_{add}+E_1+E_2=O(n(1+\frac{\log n}{r}))$.  
We must also account for elements that are marked as deleted, but are still stored in sequences $S_i$ for $i\ge 1$. The number of elements that are marked as deleted is bounded by $O(n/r)$. These elements  need $O(n\frac{\log n}{r})$ bits. Every deleted element spoils up to $k$ symbols of $S$. Using the same analysis as above, the extra encoding length due to spoiled symbols can be estimated to be $O(n(1+\frac{\log n}{r}))$. Thus all static sequences $S_i$ for $i\ge 1$ are stored in $ nH_k+ O(n(1+\frac{\log n}{r}))$ bits.

\section{Static Data Structure}
\label{sec:construct}
\tolerance=1000
In this section we describe a  static data structure supporting 
access, rank and select queries. In comparison to previous static data structures, we obtain two additional results. Our data structure can be constructed quickly if the alphabet size $\sigma$ is small.  At the same time we show that our data structure supports extraction of a  substring of length $\ell$ in optimal $O(\log n/\log\log n +\ell/\log_{\sigma}n)$ time. 
As before, let $S$ denote a sequence of length $n$ over an  alphabet $\Sigma=\{\,1\ldots,\sigma\,\}$. 

Our static representation keeps the sequence $S$ in compressed form following the approach of\cite{FerraginaV07}. $S$ is represented as a sequence $S^M$of meta-symbols
over an alphabet $\Sigma^{\ell}$ for $\ell=\ceil{\frac{\log_{\sigma}n}{2}}$. That is, each meta-symbol encodes $\ell$ symbols of the original sequence. 
It is shown in~\cite{FerraginaV07} that $H_0(S^M)\le H_k(S) + (n/\ell)k\log\sigma$ simultaneously for all $k\le\ell$. We can keep $S^M$ in $(n/\ell)(H_0(S^M)+O(1))$ bits using e.g. Huffman coding. \no{For any $k=o(\log_\sigma n)$, this representation needs $nH_k(S) + o(nH_k(S)) +O(\frac{k+1}{\log_{\sigma}n}\log \sigma) $ bits.}
\paragraph{Data Structure for Rank and Select Queries.}
 We split $S$ into blocks of size $\sigma$. For every $a\in [1,\sigma]$, we keep a binary sequence $B_a=1^{s_1}01^{s_2}0\ldots 1^{s_l}0$
where $s_i$ denotes the number of $a$'s occurrences in the $i$-th block. 
It was shown in 
\cite{BHMR07} how query $\ra_a(i,S)$ or $\sel_a(i,S)$  can be reduced to $O(1)$ rank and select queries on a block $C$ and $O(1)$ queries on $B_a$.
The data structure for a block  $C$ is as follows. We keep a bit vector $V=1^{n_1}01^{n_2}0\ldots 1^{n_{\sigma}}$ where $n_a$ is the number of times $a$ occurs in $C$. Let $\pi(i)$ denote the 
position of $C[i]$ in the stable sorted ordering. That is, $\pi$ is the permutation of $C$ 
obtained by stably sorting the symbols of $C$. Let $\pi^{-1}$ denote the inverse of $\pi$. Then $\sel_a(i,C)=\pi^{-1}(j)$ where $j=(\sum_{g=1}^{a-1}n_g) +i$. We can find $\sum_{g=1}^p n_g$ for any $p\le \sigma$ by answering one rank and one select query on $V$. 

Let $t=\log \sigma/(\log \log \sigma)^3$.
For every symbol $a$, $1\leq a\le\sigma$, the set  $F_a$ contains every $t$-th occurrence of $a$ in $C$; that is $F_a$ contains all $j$ such that $C[j]=a$ and $\ra_a(j,C)=t\cdot i$ for some integer $i$. We keep a y-trie data structure on $F_a$, so that for any $q$ we can find the largest $j\in F_a$ satisfying $j\le q$. Furthermore we store  values of $\ra_a(j,C)$ for all $j\in F_a$. For each symbol $C[j]$, we also keep $R[j]= (\ra_{C[j]}(j,C) \mod t)$. We need $\sigma\log t$ bits to store the array $R$ and $O((\sigma/t)\log \sigma)=O(\sigma(\log\log \sigma)^3)$ bits to store $F$. Hence the total space usage is $o(\sigma\log \sigma)$.

Let $\ra'_a(i,S)$ denote the partial rank query: if $S[i]=a$, $\ra_a'(i,S)=\ra_a(i,S)$; otherwise $\ra'_a(i,S)$ is undefined. If $C[i]=a$, $\ra'_a(i,C)=R[i]+\ra_a(j,C)$ where $j$ is the largest position in $F_a$ such that $j\le i$. Since $j$ can be found in $(\log\log\sigma)$ time, $\ra'_{C[i]}(i,C)$ can be computed in $O(\log\log\sigma)$ time.
We can compute $\pi(i)$ as follows. If $C[i]=a$, then $\pi(i)=(\sum_{j=1}^{a-1}n_j) + \ra'_a(i,C)$. Since $\ra'$ can be computed in $O(\log\log \sigma)$ time, we can find $\pi(i)$ for any $i$ in $O(\log\log \sigma)$ time. 
Using the data structure of~\cite{MunroRRR12}, we can compute $\pi^{-1}(i)$ in $O(t\cdot f(\sigma))$ time using $O(n\log\sigma/t)$ additional bits, where $f(\sigma)$ is the time needed to compute $f(\sigma)$. This data structure works as follows:
We decompose the permutation $\pi=\pi(1),\pi(2),\ldots, \pi(\sigma)$ into cycles. A cycle is the shortest subsequence $i_1,\ldots ,i_s$ of $\pi$ such that $\pi(i_j)=i_{j+1}$ for $1\le j< s$ and $\pi(i_s)=i_1$. For every cycle of length $s\ge t$, we select every $t$-th element and mark it. We keep the value of $\pi^{-t}$ for the marked elements where $\pi^{-t}$ denotes the inverse of $\pi$ iterated $t$ times. In order to find $\pi^{-1}(i)$, we compute $\pi(i)$, $\pi^2(i)=\pi(\pi(i))$, $\pi^3(i)$, $\ldots$ until we reach a marked position $i_m$ or $\pi^{k}(i)=i$ for some $k$. If $\pi^{k}(i)=i$, then $\pi^{-1}(i)=\pi^{k-1}(i)$. If we reached a marked position $i_m$, we compute $i'=\pi^{-t}(i_m)$. Then we identify $\pi(i')$, $\pi^2(i')$, $\ldots$ until $\pi^l(i')=i$. Clearly $\pi^{-1}(i)=\pi^{l-1}(i')$ in this case. It is easy to check that we must compute $\pi$ at most $t$ times; details can be found in~\cite{MunroRRR12}.  
Thus $\pi^{-1}(i)$ is computed in $O(t\cdot \log\log \sigma)$ time. We already showed how to answer $\sel$ query using $\pi^{-1}$. Hence $\sel_a(i,C)$ is also answered in $O(t\log\log\sigma)$ time. To answer a rank query $\ra_a(i,C)$, we first find the largest $j\in F_a$ such that $j\le i$. If $\ra_a(j,C)=s\cdot t$, then $st\le \ra_a(i,C) < (s+1)t$. We can find the exact value of $\ra_a(i,C)$ by answering $O(\log t)$ $\sel$ queries as described in\cite{GMR06,BHMR07}. Hence $\ra_a(i,C)$ is computed in $O(t\log t \log\log \sigma)$ time.  We set $t=\log \sigma/(\log \log \sigma)^3$. Hence a query $\ra_a(i,C)$ is answered in $O(\log\sigma/\log\log\sigma)$ time.

\paragraph{Linear Construction Time.}
The data structure described above can be constructed in $O(n)$ time.  We can split $S$ into blocks in linear time. Then we stably sort each block and compute the number of times $n_a$ the symbol $a\in\Sigma$ occurs in a block. We can implement stable sorting by replacing each $C[i]$ with $C'[i]=C[i]\cdot \sigma+ i$ and applying radix sort to the resulting sequence. 
Using sorted array $C'$, we can: (i) compute $\pi(i)$ for each position $i$ within a block; (ii) find values of $n_a$ for each symbol $a$ and construct the sequence $V$; (iii) generate sets $F'_a$ and the array $R$. All these auxiliary structures can be created in linear time. 
We can construct a $y$-trie for $F_a$ in $O(|F_a|(\log\log \sigma)^3)$ time: 
each element of a y-trie is kept in $O(\log\log \sigma)$ dictionary data structures; using the deterministic method described in~\cite{Ruzic08}, we can construct a dictionary with $m\le \sigma$ elements in $O(m(\log\log\sigma)^2)$ time. 
Hence the total time needed to construct a y-trie is $O(|F_a|(\log\log\sigma)^3)$. Since all $F_a$ contain $O(\sigma/t)$ elements,
y-tries for all $F_a$ are created in $O(\sigma(\log\log \sigma)^3/t +\sigma)$ time.
Since we can compute $\pi(i)$ for each $i$ in $O(1)$ time using $C'$, we can produce a data structure for computing $\pi^{-1}$ in linear time. Thus the data structure for answering rank and select queries in a block can 
be created in $O(\sigma)$ time. When values of $n_a$ are known for all blocks we can  construct global bit sequences $B_a$ for each  $a\in \Sigma$. 

\paragraph{Data Structure  for $\log^{1/2}n < \sigma\le 2^{\log^{1/3}n}$.}
In this case the data structure can be constructed in less than linear time. We assume that the symbols of $S$ are initially packed into words of $\log n$ bits so that each word contains $\Theta(\log_{\sigma}n)$ symbols. 
We split the sequence $S$ into blocks of size $s=\sigma\log n$. We keep exactly the same 
data structures for each block as in the case of  $\sigma> 2^{\log^{1/3}n}$ and bit sequences $B_a$ defined in the same way as above. 
We start by splitting $S$ into blocks and producing an array $C'$ for each block $C$ so that 
$C'[i]=C[i]\cdot s + i$. This step takes $O(s/\log^{2/3}n)$ time. 
$C'$ can be sorted in $O(n/\log^{1/3}n)$ time, using the ideas of sorting algorithms for small integers described in~\cite{AnderssonHNR95} and~\cite{AlbersH97} . Then we can traverse sorted array $C'$ and generate sets that must be stored in data structures $F_a$ in $O((s/\log^{2/3}n)+\sigma)$ time. 
All $F_a$ contain $O(s/t)$ elements and can be constructed in $O((s/t)(\log\log n)^3+ \sigma)=O((s/t)(\log\log n)^3)$ time. 
We traverse $C'$ again and obtain $R'[i]=R[C'[i]]$ for each $i$. Given $R'$, we can construct $R$ by ``reverse sorting''. Let $R_1[i]=C'[i]*(\log \sigma)+ R'[i]$. That is, the first $\log \sigma$ most significant  bits of $R_1[i]$ contain a symbol $C[j]$ of the sequence $C$, the next $\log s$ bits contain its position $j$ in $C$, the next $\log\log \sigma$ bits contain the value of $R[j]$. We sort $R_1$ according to 
the value of bits at positions $\log \sigma +1$, $\ldots$, $\log \sigma+\log s$ (bits that correspond to the positions $j$ of symbols in the original sequence) and then discard the first $\log \sigma +\log s$ bits. 
The resulting array is the array $R$. 

We can also use $C'$ to construct the bit sequence $V$: we traverse $C'$ and compute $n_a$ for all $a$, $1\le a\le \sigma$. When all $n_a$ are known, we can produce $V$ in $O(s/\log^{1/3} n)$ time; a data structure supporting rank and select on $V$ can be also produced in $O(s/\log^{1/3}n)$ time. 

Finally we need to create the data structure for computing $\pi^{-1}$. Recall that we have to find all cycles of length at least $t$ and select every $t$-th element in a cycle. 
Let $d=\log^fn$ for $f=1/6$. During the first stage we create $s/d$ tuples so that each tuple is of the form $(j,\pi(j),\pi^2(j),\ldots,\pi^r(j))$ for some $r\le d$ and each integer $i$, $1\le i\le s$ occurs in at most one tuple. 
First we obtain values $\pi(i)$ for all $i\in [1,s]$ and keep tuples 
$(i,\pi(i))$ in the array $P_1$. Using $C'$, we can obtain $P_1$ in $O(s/\log^{1/3}n)$ time. We traverse $P_1$ and remove all tuples 
$(i,\pi(i))$ such that $\pi(i)=i$. Then we obtain the sequence $P_2$ 
that contains tuples $(i,\pi(i),\pi^2(i))$ for all $i$ such that $(i,\pi(i))$ 
is still in $P_1$. We create a new instance $P_1'$ of $P_1$ and sort all tuples by their second components. Elements of $P_1'$ are tuples $(i,\pi(i))$ sorted by $\pi(i)$. Elements of $P_1$ are  tuples $(i,\pi(i))$ sorted by $i$. Both $P_1$ and $P_1'$ are traversed simultaneously. 
If the $j$-th tuple in $P_1$ is $(i_j,\pi(i_j))$ and $\pi(i_j)=v$, then the
$j$-th tuple in $P_1'$ is $(v,\pi(v))$. When we read the $P_1[j]$ and $P_1'[j]$, we create the new tuple $(i_j,\pi(i_j),\pi(v)=\pi^2(i_j))$ and keep it in a sequence $P_2$. When $P_2$ is constructed, we discard $P_1'$;
then we traverse $P_2$ and remove all $(i_k,\pi(i_k),\pi^2(i_k))$ such that 
$\pi^2(i_k)=\pi(i_k)$. This procedure is iterated $d-1$ times. During the $k$-th iteration, we sort tuples in $P_{k-1}$ by their last components and obtain $P_{k-1}'$. Then we merge $P'_{k-1}$ with $P_1$ and obtain $P_k$. 
We traverse $P_k$ and remove tuples $(i_j,\ldots,\pi^k(i_j))$ satisfying 
$\pi^k(i_j)=i_j$. Each iteration takes $O(s/\log^{1/3}n)$ time. Hence $P_d$ 
is obtained in $O(s/\log^{1/6}n)$ time.

At the end of the first stage we obtain the sequence $P_d$. Every value 
$i$ that is not in a cycle of length $v\le d$ is stored in exactly one tuple of $P_d$. Hence $P_d$ consists of $s/d$ tuples. We can easily process all tuples in $O(|P_d|)$ time and find all values $i$, $1\le i\le s$, that 
must be marked. We can find $\pi^{-t}(i)$ for all marked positions $i$ in $O(s/d)$ time. Thus the structure for computing $\pi^{-1}$ is constructed in $O(s/\log^{1/6}n)$ time. The total time needed to produce the static data structure for a sequence $S$ is thus $O(|S|/\log^{1/6}n)$.

\paragraph{Data Structure for $\sigma< \log^{1/2} n$}
In the case when $\sigma$ is very small, we use a different data structure. 
We implement rank and select operations on $S$ using the result of Theorem 13 in~\cite{BN12}. Their data structure splits $S$ into chunks of size $\log_{\sigma}n/2$. Each chunk is kept as in~\cite{RRR07}. We can traverse $S$ and obtain compressed representation of each chunk in $O(n/\log_{\sigma}n)$ time. We maintain certain bit sequences for chunks that are described in~\cite{BN12} and can be constructed in $O(n\sigma/\log_{\sigma}n)$ time. Since $\sigma<\log^{1/2}n$, $O(n\sigma/\log_{\sigma}n)=O(n\log\log n/\log^{1/2}n)$. This representation of $S$ also supports fast substring extraction: since $S$ is kept in chunks, we can decode all symbols from a chunk in $O(1)$ time and retrieve a string of length $l$ in $O(l/\log_{\sigma}n)$ time. 

\begin{theorem}
  \label{theor:static}
There exists a data structure $D$  that that stores a sequence $S[1,n]$ in  $nH_k+O(n(\log\log\sigma)^3)$ bits, where $\sigma$ is the alphabet size,  and supports queries $\acc$, $\ra$, and $\sel$ in $O(\log n/\log\log n)$ time.  $D$ can be constructed in $O(n)$ time. \\
Suppose that $\sigma\le 2^{\log^{1/3}n}$ and $S$ is initially stored in $O(n/\log_{\sigma}n)$ words, so that every word contains $\Theta(\log_{\sigma}n)$ consecutive symbols; then $D$ can be constructed in $O(n/\log^{1/6}n)$ time.  
\end{theorem}
Finally we remark about re-building static sequences that is needed by background processes described in Section~\ref{sec:updatesbackground}. When a subsequence $S_i$ is re-built, we retrieve $S_i$ using the algorithm for substring extraction in $O(|S_i|/\log_{\sigma}n)$ time. The decoded sequence $S_i$ 
is then kept in uncompressed form; we keep $S_i$ in a sequence of words, so that each word contains $\log_{\sigma}n$ symbols. We can apply the construction algorithms described in this section to uncompressed sequence $S_i$. The workspace needed to store $S_i$ in plain form is $O(|S_i|\log \sigma)$ bits.

\section{Operation $\sel'$ on $\tS$ and Reporting a Substring of a Binary Sequence}
\label{sec:appselprime}
\paragraph{Operation $\sel'_a(i,\tS)$.}
Let $S_a$ be the subsequence of $S$ that consists of all occurrences of a symbol $a$. 
We maintain a bit sequence $\tW$ for each sequence $S_a$. 
For every element of $S_a$, we keep one or two consecutive bits 
in $\tW$. If the $j$-th occurrence of $a$ is not stored in $\tS$, then we represent it by a $0$; if the $j$-th occurrence 
is stored in $\tS$ (i.e., it is stored in either $S'_a$ or $S_0$), then we represent it by 
a two-bit sequence $10$. Let $f$ denote the number of symbols in $S'_a$ among the first $j$ symbols of $S_a$. Then $S_a[j]$ is represented by the $(j+f)$-th bit in $\tW$ or by the $(j+f)$-th and  $(j+f+1)$-st bits in $\tW$: if $S_a[j]$ is stored 
in $S'_a$, then $\tW[f+j]=1$ and $\tW[f+j+1]=0$; otherwise 
$\tW[f+j]=0$ and $\tW[f+j+1]$ represents the next symbol in $S_a$. We can answer rank and select queries on $\tW$ and support updates on $\tW$ in $O(\log n/\log \log n)$ time. 
Let $v_1=\sel_0(i,\tW)$ and $v_2=\ra_1(v_1,\tW)$. Then 
$\sel'_a(i,\tS)=v_2$.

\paragraph{Reporting a Substring  in a Binary Sequence.}
Let $M$ be a binary sequence. 
We prove the following Lemma:
\begin{lemma}
  \label{lemma:substr}
Let $M$ be a binary  sequence of length $n$ with  
$O(n/r)$ $0$-bits. We can store $M$ in $O((n/r)\log r)$ bits, 
so that any substring $M[i..i+\ell-1]$ can be obtained in $O(\log n/\log\log n+\ell/\log n)$ time.
Insertions and deletions are supported in $O(\log n/\log \log n)$ time. 
\end{lemma}
\begin{proof}
  We store $M$ using a variant of run-length encoding: 
each substring that consists of $d$ $1$-bits followed by a $0$-bit, where $0\le d \le 2\log^2 n$, is encoded as an integer $d$. For instance, a sequence $100011110$ will be encoded as 
$1,0,0,0,4$. We divide the run-length encoded sequence into blocks, such that each block consists of at least $\log n/8\log\log n$ and at most $\log n/4$ run-lengths and the length of each block is at most $\log n/2$ bits. Run-lengths are delta-encoded so that a run of length $d$ uses $\log d+o(\log d)$ bits. Thus  each block contains $\Omega(\log n)$ bits. 

We  also maintain an additional data structure $A$ that finds for each position $j$ in $M$, the run-length $d$ that encodes $M[j]$ and the block that contains the run-length $d$.  $A$ encodes every run-length in unary. Thus a run of length $d$ is represented by 
$1^d0$. Since $M$ contains $O(n/r)$ $0$-bits, $M$ consists of $O(n/r)$ runs of $1$'s followed by a $0$. Hence $A$ consists of $O(n/r)$ $0$'s and $O(n)$ $1$-bits. The sequence $A'$ encodes in unary the number of runs in every block of $M$.  Using standard methods, we can keep $A$ and $A'$ in $O((n/r)\log r)$ bits and support queries and updates in $O(\log n/\log \log n)$ time. 
Using rank and select queries on $A$ and $A'$, we can find the block that encodes $M[j]$ and the position of $M[j]$ in its block
for any $j$, $1\le j\le n$, in $O(\log n/\log\log n)$ time. 
We also keep a look-up table $Tbl$ that enables us to  retrieve all $k$ elements stored in a 
block in $O(k/\log n)$ time; for every block, $Tbl$ contains the sequence of bits encoded by this block.  Since there are $O(n^{1/2})$ different blocks and each block encodes a poly-logarithmic number of elements, $Tbl$ uses $o(n)$ bits. 

 Each block contains either at least $\log^2n/4\log\log n$ $1$-bits or at least 
$\log n/4\log\log n$ $0$-bits. Hence the total number of blocks 
is $O(n (\log \log n/\log^2n) + (n/r)(\log \log n/\log n))$. 
Each block needs $O(\log n)$ bits. Hence all blocks 
use $O((n/r)\log\log n)$ bits.  

To extract a substring $M[i..i+\ell-1]$, we start by finding the block $Bl$ that contains $M[i]$ and the position of $M[i]$ in $Bl$.  Then we simply decode the remaining part of the block $Bl$ and the following blocks until $O(\ell)$ symbols are decoded.
\end{proof}

\section{Substring Extraction}
\label{sec:substr}
Now we show how the fully-dynamic data structure described in Section~\ref{sec:compr2} supports the operation of retrieving a substring of length $\ell$.  
Suppose that we want to extract the substring $S[i..i+\ell]$.
We keep a copy $S_w$ of subsequence $S_0$ implemented as follows. 
$S_w$ is split into words, such that each word contains between $\log_{\sigma}n/4$ and $\log_{\sigma}n/2$ symbols of $S_0$. Let $w_i$ be the number of symbols in the $i$-th word; 
we maintain a prefix-sum data structure on $w_i$. Using this data structure, we can find the word $S_w[j]$ that contains the $i$-th symbol of $S_0$ in $O(\log n/\log\log n)$ time. We can find the position $o_i$ of $S_0[i]$ in that word in $O(1)$ time. Using table look-up, we can extract the remaining symbols of $S_w[j]$ in $O(1)$ time. If $w_j-o_i<\ell$, we extract the following symbols from words $S_w[j+1]$, $S_w[j+2]$, $\ldots$ until $\ell$ symbols are 
reported. 

The static data structure on $S_i$ can be used to extract $\ell$ symbols in $O(\log n/\log\log n+\ell/\log_{\sigma}n)$ time. Some of these symbols can, however, be marked as deleted. We use the following additional structures in order to extract $\Theta(\log_{\sigma}n)$ undeleted symbols in $O(1)$ time.
Recall that each sequence $S_i$ is stored as a sequence of meta-symbols $S^M_i$ and every meta-symbol represents $\ceil{\frac{\log_{\sigma}n}{2}}$ symbols. We say that a meta-symbol $S^M_i[j]$ is \emph{spoiled} if at least $\log_{\sigma}n/4$ symbols represented by $S^M_i[j]$ are marked as deleted. A symbol is spoiled if it is stored in a spoiled meta-symbol. Positions of spoiled symbols are indicated by a binary sequence $SPOS_i$. That is, $SPOS_i[j]=1$ iff the symbol $S_i[j]$ is not spoiled. Symbols stored in spoiled meta-symbols are also kept in a sequence $V_i$. Representation of  $V_i$ is similar to representation of $S_w$, but it contains only undeleted symbols stored in spoiled meta-symbols. $V_i$ is divided into words and each word $V_i^M[j]$ contains up to $\log_{\sigma}n/2$ 
symbols.  If a word $V_i^M[j]$ contains less than $\log_{\sigma}n/4$ symbols, than the last symbol in this word is followed by a non-spoiled symbol. Each word is augmented with a field $next$. 
Let $fol(j)$ denote the symbol that follows the last symbol in $V_i^M[j]$.
 $V^M_i[j].next=NULL$ if $fol(j)$ is spoiled; otherwise $V^M_i[j].next$ points to the position of $fol(j)$ in $S_i^M$. 
A sequence  $VPOS_i$ indicates boundaries of words in $V_i$: $VPOS_i$ contains a $0$-bit for every symbol in $V_i$ that is not the last symbol in 
its word $V_i^M$; $VPOS_i$ contains a two-bit substring $01$ for every symbol that is the last symbol in its word. Thus each symbol is encoded by a $0$-bit and the end of every word in $V$ is encoded by a $1$-bit. If a symbol $S_i[j]$ is not marked as deleted and kept in a spoiled meta-symbol, then we can find the position of $S_i[j]$ in $V_i$ by answering one rank query on $SPOS_i$ and one rank and one select query on $VPOS_i$. 

The total number of symbols that are marked as deleted in all $S_i$ is bounded by $O(n/r)$. Hence the number of spoiled symbols in all $S_i$ is also $O(n/r)$. 
Non-deleted symbols kept in a spoiled meta-symbol are stored in at most three words of $V_i^M$. Hence the total number of words in all $V_i^M$ is bounded by $O(n/(r\log_{\sigma}n))$. Since every word uses $O(\log n)$ bits of space, all $V_i$ need $O((n/r)\log \sigma)$  bits.  All bit sequences $VPOS_i$ and $SPOS_i$ use $O(n/r)$ and $O((n/r)\log r)$ bits respectively. Hence we need 
$O((n/r)\log\sigma)$ additional bits in order to support substring extraction.

Suppose that a string $S_t[i..i+l]$ must be extracted. We find the meta-symbol $S^M_t[j_0]$ that contains $S_t[i]$ and decode meta-symbols $S^M_t[j_0]$, $S^M_t[j_0+1]$, $\ldots$ and output the appropriate symbols until $l$ symbols are reported or a spoiled meta-symbol is encountered. If the symbol $S^M_t[j]$ 
is spoiled, we find the position of $S^M_t[j]$ in $V_t$ and output  symbols from $V_t$. If we enumerated symbols of $V^M_t[j_1]$ and $V^M_t[j_1].next\not=Null$, then we switch back to $S^M_t[j_2]$, where $S^M_t[j_2]$ is the meta-symbol that is pointed to by $V^M_t[j_1].next$, and decode symbols from $S^M_t[j_2]$, $S^M_t[j_2+1]$, $\ldots$ until a spoiled symbol is encountered. We output symbols from $V_t^M[j_1+1]$, $\ldots$, $V^M_t[j_2]$ until $V_t^M[j_2].next\not=NULL$.  
We proceed in the same way until $l$ symbols are decoded. 
Each meta-symbol of $S_t$ and each word of $V_t$ is processed in $O(1)$ time. It is easy to check that the total number of words and meta-symbols  is bounded by $O(l/\log_{\sigma} n)$. Every retrieved non-spoiled symbol in $S_t^M$, except for the first one and the last one, contains $\Theta(\log_{\sigma}n)$ symbols. Every processed word in $V^M_t$, except for the last one, either contains $\Theta(\log_{\sigma}n)$ symbols or is followed by a non-spoiled meta-symbol.  The position of the first accessed spoiled symbol in $V_t$ is computed in $O(\log n/\log\log n)$ time. The position of the first accessed meta-symbol in $S_t^M$ is also computed in $O(\log n/\log\log n)$ time.  Thus the total query time is $O(\log n/\log \log n +l/\log_{\sigma}n)$.

 The extraction of $l$ symbols $S[i..i+l]$ from the global sequence $S$ is implemented as follows. We find $i_0=\ra_0(i,R)$ and $i_1=\ra_1(i,R)$. We compute  $t$ such that $\sum_{j=1}^{t-1}Size[j]< i_1\le \sum_{j=1}^jSize[j]$ and extract substring $S_t[f..f+l]$ for $f=i_1-\sum_{j=1}^{t-1}Size[j]$. If the end of $S_t$ is reached, we extract remaining symbols from $S_{t+1}$, $S_{t+2}$, $\ldots$. 
We also extract $S_0[i_0..i_0+l]$. Let $Str_1$ be the substring extracted from the static subsequence (or subsequences) and let $Str_0$ be the string extracted from $S_0$. We can merge the prefix of $Str_1$ with the prefix of $Str_0$ using $R$. 
At each step we consider the next $\log_{\sigma}n/6$ symbols of $Str_0$ and $\log_{\sigma}n/6$ symbols of $Str_1$ that are not processed yet. Suppose that these symbols are stored in words $W_0$ and $W_1$ respectively.  We read the next $\log_{\sigma}n/6$ bits of $R$ and keep them in a bit sequence $R_W$. Using a look-up table, we can obtain the sequence $W_{res}$ that consists of $\log_{\sigma}n/6$ following symbols in $O(1)$ time: if $R_W[j]=1$, then the $j$-th symbol of $W_{res}$ is the $r_0$-th symbol of $W_0$ where $r_0$ is the number of $0$'s among the first $j$ bits of $R_W$; otherwise the $j$-th symbol of $W_{res}$ is the $r_1$-th symbol of $W_1$ where $r_1$ is the number of $1$'s among the first $j$ bits of $R_W$. The sequence $W_{res}$ contains the next $\log_{\sigma}n/6$ symbols of $S[i..i+l]$. Proceeding in the same way, we can obtain the substring $S[i..i+l]$ in $O(6l/\log_{\sigma}n)=O(l/\log_{\sigma}n)$ time.

\end{document}